\documentclass{article}
\usepackage{fullpage}
\usepackage[latin2]{inputenc}
\usepackage[english]{babel}
\usepackage[cmex10]{amsmath}
\usepackage[all=normal,paragraphs=tight,bibliography=tight]{savetrees}
\usepackage{todonotes}
\usepackage{amsthm}
\usepackage{amssymb}
\usepackage{microtype}
\usepackage{xspace}
\usepackage{todonotes}
\usepackage{comment}
\usepackage{enumerate}
\usepackage{tikz}
\usepackage{blindtext}

\usepackage{graphicx}
\usepackage{caption}
\usepackage{subcaption}
\def\cqedsymbol{\ifmmode$\lrcorner$\else{\unskip\nobreak\hfil
\penalty50\hskip1em\null\nobreak\hfil$\lrcorner$
\parfillskip=0pt\finalhyphendemerits=0\endgraf}\fi} 

\newcommand{\cqed}{\renewcommand{\qed}{\cqedsymbol}}

\numberwithin{equation}{section}
\newtheorem{theorem}{Theorem}[section]
\newtheorem{lemma}[theorem]{Lemma}
\newtheorem{claim}[theorem]{Claim}

\theoremstyle{definition}
\newtheorem{definition}[theorem]{Definition}

\newcommand{\Oh}{\mathcal{O}}
\newcommand{\Ohstar}{\Oh^\star}

\newcommand{\picname}{\textsc{Proper Interval Completion}\xspace}

\title{A subexponential parameterized algorithm for\\ \textsc{Proper Interval Completion}%
\thanks{The research leading to these results has received funding from the European Research Council under the European Union's Seventh Framework Programme (FP/2007-2013) / ERC Grant Agreement n. 267959}}
\author{
  Ivan Bliznets%
  \thanks{
    St. Petersburg Academic University of the Russian Academy of Sciences, Russia, \texttt{ivanbliznets@tut.by}.
  }
  \and
  Fedor V. Fomin%
  \thanks{
    Department of Informatics, University of Bergen, Norway, \texttt{fomin@ii.uib.no}.
  }
  \and
  Marcin Pilipczuk%
  \thanks{
    Department of Informatics, University of Bergen, Norway, \texttt{Marcin.Pilipczuk@ii.uib.no}.
  }
  \and
  Micha\l{} Pilipczuk%
  \thanks{
    Department of Informatics, University of Bergen, Norway, \texttt{michal.pilipczuk@ii.uib.no}.
  }
  }

\date{}

\begin{document}

\maketitle

\begin{abstract}
In the {\sc{Proper Interval Completion}} problem we are given a graph $G$ and an integer $k$, and the task is to turn $G$ using at most $k$ edge additions into a proper interval graph, i.e., a graph admitting an intersection model of equal-length intervals on a line. The study of {\sc{Proper Interval Completion}} from the viewpoint of parameterized complexity has been initiated by Kaplan, Shamir and Tarjan~[FOCS 1994; SIAM J. Comput. 1999], who showed an algorithm for the problem working in $\Oh(16^k\cdot (n+m))$ time. In this paper we present an algorithm with running time $k^{\Oh(k^{2/3})} + \Oh(nm(kn+m))$, which is the first subexponential parameterized algorithm for {\sc{Proper Interval Completion}}.

\end{abstract}

\section{Introduction}\label{sec:intro}
A graph $G$ is an {\em{interval graph}} if it admits a model of the following form: each vertex is associated with an interval on the real line, and two vertices are adjacent if and only if the associated intervals overlap. If moreover the intervals can be assumed to be of equal length, then $G$ is a {\em{proper interval graph}}; equivalently, one may require that no associated interval is contained in another~\cite{roberts}. Interval and proper interval graphs appear naturally in molecular biology in the problem of {\em{physical mapping}}, where one is given a graph with vertices modelling contiguous intervals (called {\em{clones}}) in a DNA sequence, and the edges indicate which intervals overlap. Based on this information one would like to reconstruct the layout of the clones. We refer to~\cite{goldberg1995four,Golumbic80,tarjan-sicomp} for further discussion on biological applications of (proper) interval graphs.

The biological motivation was the starting point of the work of Kaplan et al.~\cite{tarjan-sicomp}, who initiated the study of (proper) interval graphs from the point of view of parameterized complexity. It is namely natural to expect that some information about overlaps will be lost, and hence the model will be missing a small number of edges. Thus we arrive at the problems of {\sc{Interval Completion}} ({\sc{IC}}) and {\sc{Proper Interval Completion}} ({\sc{PIC}}): given a graph $G$ and an integer $k$, one is asked to add at most $k$ edges to $G$ to obtain a (proper) interval graph. Both of the problems are known to be NP-hard~\cite{Yannakakis81}, and hence it is natural to ask for an FPT algorithm parameterized by the expected number of additions $k$. For {\sc{Proper Interval Completion}} Kaplan et al.~\cite{tarjan-sicomp} presented an algorithm with running time $\Oh(16^k\cdot (n+m))$, while fixed-parameterized tractability of {\sc{Interval Completion}} was resolved much later by Villanger et al.~\cite{Villanger:2009ez}.
Recently, Liu et al.  \cite{cocoon13} obtained $\Oh(4^k + nm(n + m))$-time algorithm  for  {\sc{PIC}}.

The approach of Kaplan et al.~\cite{tarjan-sicomp} is based on a characterization by {\em{forbidden induced subgraphs}}, pioneered by Cai~\cite{Cai96}: proper interval graphs are exactly graphs that are chordal, i.e., do not contain any induced cycle $C_\ell$ for $\ell\geq 4$, and moreover exclude three special structures as induced subgraphs: a {\em{claw}}, a {\em{tent}}, and a {\em{net}}. Therefore, when given a graph which is to be completed into a proper interval graph, we may apply a basic branching strategy. Whenever a forbidden induced subgraph is encountered, we branch into several possibilities of how it is going to be destroyed in the optimal solution. A cycle $C_\ell$ can be destroyed only by triangulating it, which requires adding exactly $\ell-3$ edges and can be done in roughly $4^{\ell-3}$ different ways. Since for special structures there is only a constant number of ways to destroy them, the whole branching procedure runs in $c^k n^{\Oh(1)}$ time for some constant $c$.

The approach via forbidden induced subgraphs has driven the research on the parameterized complexity of graph modification problems ever since the pioneering work of Cai~\cite{Cai96}. Of particular importance was the work on polynomial kernelization; recall that a {\em{polynomial kernel}} for a parameterized problem is a polynomial-time preprocessing routine that shrinks the size of the instance at hand to polynomial in the parameter. While many natural completion problems admit polynomial kernels, there are also examples where no polynomial kernel exists under plausible complexity assumptions~\cite{KratschW13}. In particular, {\sc{PIC}} admits a kernel with $\Oh(k^3)$ vertices which can be computed in $\Oh(nm(kn+m))$ time~\cite{pic-kernel}, while the kernelization status of {\sc{IC}} remains a notorious open problem.

The turning point came recently, when Fomin and Villanger~\cite{FominV13} proposed an algorithm for {\sc{Fill-in}}, i.e. {\sc{Chordal Completion}}, that runs in {\em{subexponential parameterized time}}, more precisely $k^{\Oh(\sqrt{k})} n^{\Oh(1)}$. As observed by Kaplan et al.~\cite{tarjan-sicomp}, the approach via forbidden induced subgraphs leads to an FPT algorithm for {\sc{Fill-in}} with running time $ 16^k n^{\Oh(1)}$. Observe that in order to achieve a subexponential running time one needs to completely abandon this route, as even branching on encountered obstacles as small as, say, induced $C_4$-s, leads to running time at least $2^k n^{\Oh(1)}$. To circumvent this, Fomin and Villanger proposed the approach of gradually building the structure of a chordal graph in a dynamic programming manner. The crucial observation was that the number of `building blocks' (in their case, {\em{potential maximal cliques}}) is subexponential in a YES-instance, and thus the dynamic program operates on a subexponential space of states.

This research direction  
 was continued by Ghosh et al.~\cite{ghosh2012faster} and by Drange et al.~\cite{DrangeFPV13}, who identified several more graph classes for which completion problems have subexponential parameterized complexity: threshold graphs, split graphs, pseudo-split graphs, and trivially perfect graphs (we refer to~\cite{DrangeFPV13,ghosh2012faster} for respective definitions). Let us remark  that problems admitting subexponential parameterized algorithms are very scarce, since for most natural parameterized problems existence of such algorithms can be refuted under the Exponential Time Hypothesis (ETH) \cite{ImpagliazzoPZ01}. Up to very recently, the only natural positive examples were problems on specifically constrained inputs, like $H$-minor free graphs~\cite{demaine2005subexponential} or tournaments~\cite{alon2009fast}. Thus, completion problems admitting subexponential parameterized algorithms can be regarded as `singular points on the complexity landscape'. Indeed, Drange et al.~\cite{DrangeFPV13} complemented their work with a number of lower bounds excluding (under ETH) subexponential parameterized algorithms for completion problems to related graphs classes, like for instance cographs.

Interestingly, threshold graphs, trivially perfect graphs and chordal graphs, which are currently our main examples, correspond to graph parameters {\em{vertex cover}}, {\em{treedepth}}, and {\em{treewidth}} in the following sense: the parameter is equal to the minimum possible maximum clique size in a completion to the graph class~($\pm 1$), see Fig.~\ref{fig:diagram}. It is therefore natural to ask if {\sc{Interval Completion}} and {\sc{Proper Interval Completion}}, which likewise correspond to {\em{pathwidth}} and {\em{bandwidth}}, also admit subexponential parameterized algorithms.

\begin{figure}
\centering
\includegraphics{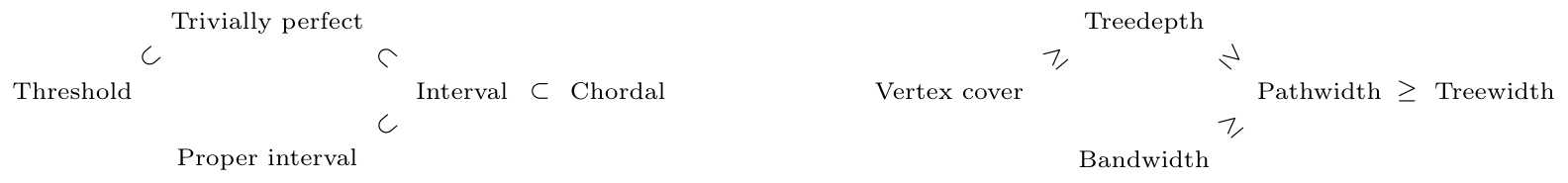}
\caption{Graph classes and corresponding graph parameters. Inequalities on the right side are with $\pm 1$ slackness.}
\label{fig:diagram}
\end{figure}

\vskip 0.3cm

\noindent {\bf{Our Results.}} In this paper we answer the question about \picname in affirmative by proving the following theorem:

\begin{theorem}\label{thm:main}
\picname can be solved in $k^{\Oh(k^{2/3})} + \Oh(nm(kn+m))$ time.
\end{theorem}

In a companion paper~\cite{my-ic} we also present an algorithm for {\sc{Interval Completion}} with running time $k^{\Oh(\sqrt{k})} n^{\Oh(1)}$, which means that the completion problems for all the classes depicted on Figure~\ref{fig:diagram} in fact do admit subexponential parameterized algorithms. We now describe briefly our techniques employed to prove Theorem~\ref{thm:main}, and main differences with the work on interval graphs~\cite{my-ic}.

From a space-level perspective, both the approach of this paper and of~\cite{my-ic} follows the route laid out by Fomin and Villanger in~\cite{FominV13}. That is, we enumerate a subexponential family of potentially interesting building blocks, and then try to arrange them into a (proper) interval model with a small number of missing edges using dynamic programming. In both cases, a natural candidate for this building block is the concept of a {\em{cut}}: given an interval model of a graph, imagine a vertical line placed at some position $x$ that pins down intervals containing $x$. A {\em{potential cut}} is then a subset of vertices that becomes a cut in some minimal completion to a (proper) interval graph of cost at most $k$. The starting point of both this work and of~\cite{my-ic} is enumeration of potential cuts. Using different structural insights into the classes of interval and proper interval graphs, one can show that in both cases the number of potential cuts is at most $n^{\Oh(\sqrt{k})}$, and they can be enumerated efficiently. Since in the case of proper interval graphs we can start with a cubic kernel given by Bessy and Perez~\cite{pic-kernel}, this immediately gives $k^{\Oh(\sqrt{k})}$ potential cuts for the {\sc{PIC}} problem. In the interval case the question of existence of a polynomial kernel is widely open, and the need of circumventing this obstacle causes severe complications in~\cite{my-ic}.

Afterwards the approaches diverge completely, as it turns out that in both cases the potential cuts are insufficient building blocks to perform dynamic programming, however for very different reasons. For {\sc{Interval Completion}} the problem is that the cut itself does not define what lies on the left and on the right of it. Even worse, there can be an exponential number of possible left/right alignments when the graph contains many modules that neighbour the same clique. To cope with this problem, the approach taken in~\cite{my-ic} remodels the dynamic programming routine so that, in some sense, the choice of left/right alignment is taken care of inside the dynamic program. The dynamic programming routine becomes thus much more complicated, and a lot of work needs to be put into bounding the number of its states, which can be very roughly viewed as quadruples of cuts enriched with an `atomic' left/right choice (see the definition of a {\em{nested terrace}} in~\cite{my-ic}). 

Curiously, in the proper interval setting the left/right choice can be easily guessed along with a potential cut at basically no extra cost. Hence, the issue causing the most severe problems in the interval case is simply non-existent. The problem, however, is in the {\em{order}} of intervals in the cut: while performing a natural left-to-right dynamic program that builds the model, we would need to ensure that intervals participating in a cut begin in the same order as they end. Therefore, apart from the cut itself and a partition of the other vertices into left and right, we would need to include in a state also the order of the vertices of the cut; as the cut may be very large, we cannot afford constructing a state for every possible order.

Instead we remodel the dynamic program, this time by introducing two layers. We first observe that the troublesome order may be guessed expeditiously providing that the cut in question has only a sublinear in $k$ number of incident edge additions. Hence, in the first layer of dynamic programming we aim at chopping the optimally completed model using such cheap cuts, and to conclude the algorithm we just need to be able to compute the best possible completed model between two border cuts that are cheap, assuming that all the intermediate cuts are expensive. This task is performed by the layer-two dynamic program. The main observation is that since all the intermediate cuts are expensive, there cannot be many disjoint such cuts and consequently the space between the border cuts is in some sense `short'. As the border cuts can be large, it is natural to start partitioning the space in between `horizontally' instead of `vertically' --- shortness of this space guarantees that the number of sensible `horizontal' separations is subexponential. The horizontal partitioning method that we employ resembles the classic $\Ohstar(10^n)$ exact algorithm for bandwidth of Feige~\cite{Feige00}.

\section{Preliminaries}\label{sec:prelims}
\providecommand{\ord}{\sigma}
\providecommand{\sol}{F}
\providecommand{\incsol}[2]{#1(#2)}
\providecommand{\incF}[1]{\incsol{\sol}{#1}}
\providecommand{\Gdown}{G_\downarrow}
\providecommand{\Gup}{G_\uparrow}
\providecommand{\spic}{\textsc{SPIC}}
\providecommand{\spicname}{\textsc{Sandwich Proper Interval Completion}}
\providecommand{\cost}{c}
\providecommand{\pos}{\Sigma}
\providecommand{\posmap}{\pi}

\paragraph{Graph notation.} In most cases, we follow standard graph notation.

An \emph{ordering} of a vertex set of a graph $G$
is a bijection $\ord: V(G) \to \{1,2,\ldots,|V(G)|\}$.
We say that a vertex $v$ is \emph{to the left} or \emph{before}
a vertex $w$ if $\ord(v) < \ord(w)$ and \emph{to the right} or \emph{after} $w$
if $\ord(v) > \ord(w)$.
We also extend these notions to orderings of subsets of vertices:
for any $X \subseteq V(G)$, any injective function $\ord: X \to \{1,2,\ldots,|V(G)|\}$
is called an ordering. We sometimes treat such $\ord$ as an ordering
of the vertex set of $G[X]$ as well, implicitly identifying $\ord(X)$ with $\{1,2,\ldots,|X|\}$ in the
monotonous way.

For any graph $G$ we shall speak about, we implicitly
fix one arbitrary ordering $\ord_0$ on $V(G)$.
We shall use this ordering to break ties and canonize some objects
(orderings, completion sets, solutions, etc.).
That is,
     assume that $X = \{x_1,x_2,\ldots,x_{|X|}\} \subseteq V(G)$
     with $\ord_0(x_1) < \ord_0(x_2) < \ldots < \ord_0(x_{|X|})$.
     Then with every ordering $\ord: X \to \{1,2,\ldots,|V(G)|\}$
we associate a sequence
$(\ord(x_1),\ord(x_2), \ldots, \ord(x_{|X|}))$,
and sort the orderings of $X$ according to this sequence lexicographically.
In many places we consider some family of orderings for a fixed choice of $X$;
if we pick the lexicographically minimum ordering of this family, we mean the one with
lexicographically minimum associated sequence.

Observe that an ordering $\ord$ of $V(G)$ naturally defines a graph $\ord(G)$ with vertex set $\{1,2,\ldots,|V(G)|\}$
and $pq \in E(\ord(G))$ if and only if $\ord^{-1}(p)\ord^{-1}(q) \in E(G)$. Clearly, $\ord(G)$ and $G$ are isomorphic
with $\ord$ being an isomorphism between them.

For any integers $a,b$ we denote $[a,b] = \{a,a+1,\ldots,b\}$.

We use $n$ and $m$ to denote the number of vertices and edges of the input graph.

\paragraph{Proper interval graphs.}
A graph $G$ is a \emph{proper interval graph} if it admits an intersection model, where
each vertex is assigned a closed interval on a line such that no interval is a proper subset of another one,
and two vertices are adjacent if and only if their intervals intersect.
In our work it is more convenient to use an equivalent combinatorial object,
called an \emph{umbrella ordering}.

\begin{definition}[umbrella ordering]
Let $G$ be a graph and $\ord: V(G) \to \{1,2,\ldots,n\}$ be an ordering of its vertices. We say
that $\ord$ satisfies the \emph{umbrella property} for a triple $a,b,c \in V(G)$
if $ac \in E(G)$ and $\ord(a) < \ord(b) < \ord(c)$ implies $ab,bc \in E(G)$.
Furthermore, $\ord$ is called an \emph{umbrella ordering} if
it satisfies the umbrella property for any $a,b,c \in V(G)$.
\end{definition}
The following result is due to Looges and Olariu. 
\begin{theorem}[\cite{umbrella}]
A graph is a proper interval graph if and only if it admits an umbrella ordering.
\end{theorem}

Observe that we may equivalently define an umbrella ordering $\ord$ as an ordering such that
for every $ab \in E(G)$ with $\ord(a) < \ord(b)$ the subgraph $\ord(G)[[a,b]]$ is a complete graph. Alternatively,  $\ord$ is an umbrella ordering of $G$ if and only if for any $a,a',b',b\in V(G)$ such that $\ord(a)\leq \ord(a')<\ord(b')\leq \ord(b)$ and $ab\in E(G)$, it also holds that $a'b'\in E(G)$. We will use these alternative definitions implicitly in the sequel.
See also Fig.~\ref{fig:umbrella} for an illustration.
\begin{figure}
\centering
\includegraphics{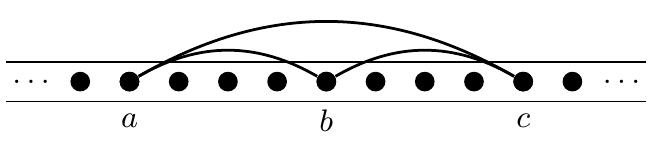}
\caption{An umbrella property for triple $a$, $b$, $c$. The existence of an edge $ac$
  implies the existence of edges $ab$ and $bc$.}
\label{fig:umbrella}
\end{figure}

Observe also the following simple fact that follows immediately from the definition of an umbrella ordering.

\begin{lemma}\label{lem:intersection-union}
Let $G_1,G_2$ be two proper interval graphs with $V(G_1)=V(G_2)=V$. Assume further that some ordering $\sigma$ of $V$ is an umbrella ordering of both $G_1$ and $G_2$. Then $\sigma$ is also an umbrella ordering of $H_\downarrow:=(V,E(G_1)\cap E(G_2))$ and $H_\uparrow:=(V,E(G_1)\cup E(G_2))$, and in particular $H_\downarrow$ and $H_\uparrow$ are proper interval graphs.
\end{lemma}

We use the assumed fixed ordering $\ord_0$ to canonize umbrella orderings:
for a proper interval graph $G$, the canonical umbrella ordering of $G$
is the one with its associated sequence being lexicographically minimum.

\paragraph{Proper interval completion.}
For a graph $G$, a \emph{completion} of $G$ is a set $\sol \subseteq \binom{V(G)}{2} \setminus E(G)$ such that $G+\sol := (V(G),E(G) \cup \sol)$ is a proper interval graph.
The \picname{} problem asks for a completion of $G$ of size not exceeding a given budget $k$.

However, in our paper it is more convenient to work with orderings as a basic notion, instead of completions.
Moreover, for technical reasons, we also need a slightly more general \emph{sandwich} version of the \picname{} problem,
henceforth called \spicname{} (\spic{} for short).
Here, apart from a graph $G$ and budget $k$, we are given
\begin{enumerate}
\item for each $u \in V(G)$ a set of \emph{allowed positions} $\pos_u \subseteq \{1,2,\ldots,V(G)\}$;
\item two graphs $\Gdown$ and $\Gup$ with vertex set $\{1,2,\ldots,|V(G)|\}$ satisfying
\begin{enumerate}
\item $\Gdown$ is a subgraph of $\Gup$;
\item both $\Gdown$ and $\Gup$ are proper interval graphs, and the identity is an umbrella ordering for both of them.
\end{enumerate}
\end{enumerate}
The \spicname{} problem asks for a completion $\sol$ of $G$, together with an ordering $\ord$ of $V(G)$, such that
\begin{enumerate}
\item $\ord$ is an umbrella ordering of $G+\sol$;
\item $\ord(u) \in \pos_u$ for each $u \in V(G)$;
\item $E(\Gdown) \subseteq E(\ord(G+\sol)) \subseteq E(\Gup)$;
\item the \emph{cost} of the ordering $\ord$ and completion $\sol$, defined as $\cost(\ord,\sol)=|\sol|$, is at most $k$.
\end{enumerate}

We now observe that an ordering $\ord$ in fact yields a unique `best' completion $\sol$.
Formally,
  for any ordering $\ord$ of $V(G)$ we define $\sol^\ord$ to be the set of such unordered pairs $xy \notin E(G)$ for which one of the following holds:
\begin{enumerate}
\item $\ord(x)\ord(y) \in E(\Gdown)$; or\label{p:sol-Gdown}
\item there exist $x',y' \in V(G)$ such that $x'y' \in E(G)$ and $\ord(x') \leq \min(\ord(x),\ord(y)) \leq \max(\ord(x),\ord(y)) \leq \ord(y')$.\label{p:sol-cap}
\end{enumerate}
We need he following property of $\sol^\ord$.
\begin{lemma}\label{lem:unique-sandwich-sol}
Set $\sol^\ord$ is a completion of $G$, $\ord$ is an umbrella ordering of $G+\sol^\ord$, and $\Gdown$ is a subgraph of $\ord(G+F^{\ord})$.
Furthermore, $\sol^\ord$ is the unique inclusion-wise minimal completion of $G$ for which $\ord$ is an umbrella ordering of $G+\sol^\ord$
and $\Gdown$ is a subgraph of $\ord(G+\sol^\ord)$.
\end{lemma}
\begin{proof}
The claim that $\Gdown$ is a subgraph of $\ord(G)$ is straightforward from the definition, as we explicitely add the edges of $E(\Gdown)$.
We now show that $\ord$ is an umbrella ordering of $G+\sol^\ord$. To this end, consider a triple $a,b,c \in V(G)$ with $\ord(a) < \ord(b) < \ord(c)$
and $ac \in E(G + \sol^\ord)$. We consider three cases, depending on the reason why $ac \in E(G+\sol^\ord)$.

If $ac \in E(G)$ then, by the second criterion of belonging to $\sol^\ord$, we have that $ab \in \sol^\ord$ unless $ab \in E(G)$, and $bc \in \sol^\ord$ unless $bc \in E(G)$.
Similarly, if $ac \in \sol^\ord$ because of the second criterion for belonging to $\sol^\ord$, then there exist   $a',c' \in V(G)$ with $a'c'\in E(G)$ and
$\ord(a') \leq \ord(a) < \ord(c) \leq \ord(c')$; clearly $a',c'$ also witness that $ab,bc \in E(G) \cup \sol^\ord$.
Finally, if $\ord(a)\ord(c) \in E(\Gdown)$, then the assumption that $\Gdown$ is a proper interval graph with identity being an umbrella ordering implies
that $\ord(a)\ord(b) \in E(\Gdown)$ and $\ord(b)\ord(c) \in E(\Gdown)$. Consequently, the umbrella property is satisfied for the triple $a,b,c$, and
$\ord$ is an umbrella ordering for $G+\sol^\ord$.

To show the second claim of the lemma, simply observe that every completion $\sol$ of $G$ for which $\Gdown$  is a subgraph of $\ord(G+F^{\ord})$   contains the edges
of $\sol^\ord$  falling into the first criterion, whereas every completion $\sol$ of $G$ for which $\ord$ is an umbrella ordering  contains the edges
of $\sol^\ord$ that fall into the second criterion.
\end{proof}
Hence, Lemma~\ref{lem:unique-sandwich-sol} allows us to use the notion of the  cost of
an ordering $\ord$ (instead of the  cost of a pair $(\ord,\sol)$ or completion $\sol$), where we use
the completion $\sol^\ord$. That is, we denote $\cost(\ord) = |\sol^\ord|$.

We say that an ordering $\ord$ is \emph{feasible} if $\ord(u) \in \pos_u$ for each $u \in V(G)$ and additionally
$E(\ord(G)) \subseteq E(\Gup)$. It is straightforward to verify using Lemma~\ref{lem:intersection-union}, minimality of $\sol^\ord$, and the fact that $\ord$ is an umbrella ordering of $\Gup$, that
the second condition for $\ord$ being feasible is equivalent to $E(\ord(G+\sol^\ord)) \subseteq E(\Gup)$.
Hence, by Lemma~\ref{lem:unique-sandwich-sol}, the \spic{} problem may equivalently ask for a feasible ordering $\ord$ of cost at most $k$.

Finally, observe that \spic{} is a generalization of \picname{}, as we may take $\pos_u = \{1,2,\ldots,|V(G)|\}$ for each $u \in V(G)$,
  $\Gdown$ to be edgeless and $\Gup$ to be a complete graph.
Note that for such an instance, any ordering of $V(G)$ is feasible.
In this way, given a \picname{} instance $(G,k)$
and an ordering $\ord$ of $V(G)$, the notions of $\sol^\ord$ and $\cost(\ord)$ are well-defined.
Hence, the \picname{} problem equivalently asks for an ordering $\ord$ of
cost at most $k$, that is, for which $|\sol^\ord| \leq k$.

We now set up a few more notions.
For a completion $\sol$ of $G$ and a vertex $v \in V(G)$ by $\incF{v}$ we denote the set of edges
$e \in \sol$ that are incident with $v$.
We extend this notion to vertex sets $X \subseteq V(G)$ by $\incF{X} = \bigcup_{v \in X} \incF{v}$.

For a \spic{} instance $(G,k,(\pos_u)_{u \in V(G)},\Gdown,\Gup)$ and a feasible ordering $\ord$
we denote $G^\ord := G+\sol^\ord$.
We extend the notion of feasibility and of $\sol^\ord$ to orderings $\ord$ of subsets of $V(G)$ in the following natural manner.
If $X \subseteq V(G)$ and $\ord: X \to \{1,2,\ldots,|V(G)|\}$ is injective, then $\ord$ is feasible if and only if $\ord(u)\in \pos_u$ for each $u\in X$ and $E(\ord(G[X]))\subseteq E(\Gup)$. The set $\sol^\ord$ is defined as follows:
$xy \in \sol^\ord$ if and only if $x,y \in X$, $xy \notin E(G)$, but either $\ord(x)\ord(y) \in E(\Gdown)$
or there exists an edge $x'y' \in E(G[X])$
with $\ord(x') \leq \min(\ord(x), \ord(y)) < \max(\ord(x),\ord(y)) \leq \ord(y')$. Again, the same argument shows that the second condition of feasibility is equivalent to $E(\ord(G[X]+\sol^\ord))\subseteq E(\Gup)$.

We use the assumed fixed ordering $\ord_0$ to canonize a solution of a \spic{} instance $(G,k,(\pos_u)_{u \in V(G)},\Gdown,\Gup)$. 
An ordering $\ord$ of $V(G)$ is called the \emph{canonical umbrella ordering} of
$(G,k,(\pos_u)_{u \in V(G)},\Gdown,\Gup)$
if $\ord$ is feasible, its cost is minimum possible, and $\ord$ is lexicographically smallest with this property.
This notion projects to the notion of a canonical umbrella ordering of a graph $G$ by taking again $\pos_u = \{1,2,\ldots,n\}$ for any $u \in V(G)$, $\Gdown$ to be edgeless and $\Gup$ to be a complete graph.
Observe that this notion thus extends the notion of canonical umbrella ordering for proper interval
graphs, as in the case of a proper interval graph the unique minimum completion is empty.

The associated completion $\sol^\ord$ with the canonical umbrella ordering $\ord$ is called the \emph{canonical completion}.
If additionally the cost of $\ord$ is at most $k$, we call $\ord$ the \emph{canonical solution}
to the \spic{} instance $(G,k,(\pos_u)_{u \in V(G)},\Gdown,\Gup)$, or, in the special case, to a \picname{} instance $(G,k)$.

\paragraph{A polynomial kernel.} Our starting point for the proof of Theorem~\ref{thm:main}
is the polynomial kernel for \picname{} due to Bessy and Perez. 

\begin{theorem}[\cite{pic-kernel}]\label{thm:kernel}
\picname{} admits a kernel with $\Oh(k^3)$ vertices computable in time $\Oh(nm(kn+m))$.
\end{theorem}
That is, in time $\Oh(nm(kn+m))$  we can  construct an equivalent instance of   \picname{} with $\Oh(k^3)$ vertices.

The algorithm of Theorem~\ref{thm:main} starts with applying the kernelization
algorithm of Theorem~\ref{thm:kernel}. This step contributes $\Oh(nm(kn+m))$ to the running
time, and all further computation will take $k^{\Oh(k^{2/3})}$ time, yielding
the promised time bound. 
Hence, in the rest of the paper we assume that we are given a \picname{} instance $(G,k)$
with $n = |V(G)| = \Oh(k^3)$, and we are targeting at the canonical umbrella ordering
of $G$ provided that it yields a completion of size at most $k$.
Moreover, we assume that $G$ is connected, as we may otherwise solve each connected component
of $G$ independently, determining in each component the size of minimum possible solution.

\paragraph{Lexicographically minimum perfect matching.}
In a few places we need the following greedy procedure to find some canonical object.
\begin{lemma}\label{lem:lex-match}
Given two linearly ordered sets $X = \{x_1 \prec x_2  \prec \cdots  \prec x_s\}$ and $Y = \{y_1 \prec y_2  \prec \cdots  \prec y_s\}$,
and allowed sets $A_i \subseteq Y$ for each $1 \leq i \leq s$,
one can in polynomial time either 
find a bijection $f: X \to Y$ that satisfies
\begin{equation}\label{eq:lex-match-cond}
f(x_i) \in A_i\qquad\mathrm{for\ any\ }1 \leq i \leq s
\end{equation}
and, subject to~\eqref{eq:lex-match-cond},
yields lexicographically minimum sequence $(f(x_1),f(x_2),\ldots,f(x_s))$, or correctly conclude that such a bijection does not exist.
\end{lemma}
\begin{proof}
We model the task of satisfying the condition~\eqref{eq:lex-match-cond}
as a problem of finding a perfect matching in a bipartite graph, which can be solved in polynomial time.
We construct an auxiliary bipartite graph $H$ with bipartition classes $X$ and $Y$, and make each $x_i \in X$ adjacent to all $y_j \in A_i$.
Clearly, any perfect matching in $H$ corresponds to a bijection $f$ satisfying~\eqref{eq:lex-match-cond}.

To obtain the lexicographically minimum sequence $(f(x_1),f(x_2),\ldots,f(x_s))$, we use the self-reducibility of the task of finding a perfect matching.
That is, for each $i=1,2,\ldots,s$ we try to match $x_i$. 
When we consider $x_i$, we try each $j=1,2,\ldots,s$ and, whenever $y_j$ is yet unmatched and $y_j \in A_i$, we temporarily match $x_i$ with $y_j$ and compute whether the subgraph induced by the currently unmatched vertices contains a perfect matching.
If this is true, we fix the match $f(x_i)=y_j$, and otherwise we proceed to the next vertex $y_j$.
It is straightforward to verify that this procedure indeed yields $f$ as desired.
\end{proof}

\section{Expensive vertices}\label{sec:exp}
Recall that we are given a \picname{} instance $(G,k)$
and we want to reason about its canonical umbrella ordering, denoted $\ord$, provided
that $(G,k)$ is a YES-instance.
In this section we deal with vertices that are incident with many edges of $\sol^\ord$.
Formally, we set a threshold $\tau := (2k)^{1/3}$ and say that a vertex $v$
is \emph{expensive} with respect to $\ord$ if $|\incsol{\sol^\ord}{v}| > \tau$,
and \emph{cheap} otherwise.
Note that there are at most $(2k)^{2/3} = \tau^2$ expensive vertices, and given that
$|V(G)|$ is bounded polynomially in $k$, we may afford guessing a lot of information about
expensive vertices within the promised time bound.
Our goal is to get rid of expensive vertices, at the cost of turning
our \picname{} instance $(G,k)$ into a \spic{} instance.

More formally, we branch into $k^{\Oh(k/\tau)} = k^{\Oh(k^{2/3})}$ subcases, considering all possible
values for the following (see also Figure~\ref{fig:expensive}).
\begin{enumerate}
\item A set $V_\$ \subseteq V(G)$ of all expensive vertices with respect to $\ord$.
\item For every $v \in V_\$$, integers $p_v$, $p_v^L$ and $p_v^R$ satisfying
$p_v = \ord(v)$, $p_v^L = \min \{\ord(w) : w \in N_{G^\ord}[v]\}$
and $p_v^R = \max \{\ord(w): w \in N_{G^\ord}[v]\}$.
\end{enumerate}
In each branch, we look for the canonical minimum solution to the instance $(G,k)$,
assuming that the aforementioned guess is a correct one.
The \emph{correct branch} is the one where this assumption is indeed true.

\begin{figure}
\centering
\includegraphics{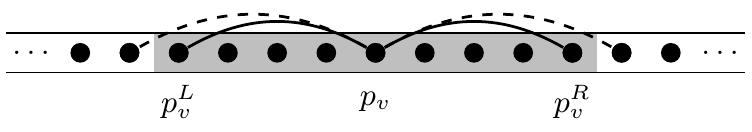}
\caption{The definition of values $p_v$, $p_v^L$ and $p_v^R$ for an expensive vertex $v$.
  The gray area denotes $N_{G^\ord}[v]$.}
\label{fig:expensive}
\end{figure}

We now perform some cleanup operations.
First, observe that from the definition of an umbrella ordering it follows
that in the correct branch
$w \in N_{G^\ord}[v]$ if and only if $p_v^L \leq \ord(w) \leq p_v^R$.
In particular, $p_v^L \leq p_v \leq p_v^R$.
Consider now a pair $v_1,v_2 \in V_\$$ and observe the following.
If $p_{v_1} \leq p_{v_2}$ then the properties of an umbrella ordering implies that
$p_{v_1}^L \leq p_{v_2}^L$ and
$p_{v_1}^R \leq p_{v_2}^R$.
Hence, we terminate all the branches where any of these inequalities is not satisfied,
or where $p_{v_1} = p_{v_2}$ for some $v_1 \neq v_2$.

Furthermore, note that in the correct branch we have $v_1v_2 \in E(G^\ord)$ iff
$p_{v_2} \in [p_{v_1}^L, p_{v_1}^R]$ and $p_{v_1} \in [p_{v_2}^L, p_{v_2}^R]$,
and $v_1v_2 \notin E(G^\ord)$ iff \emph{neither} of the two aforementioned
inclusions hold. Thus, we terminate the branch if \emph{exactly one}
of these inclusions holds, or if $v_1v_2 \in E(G)$ and at least one of them
does not hold.

Denote $\pos_\$ = \{p_v: v \in V_\$\}$ to be the set of positions guessed to be used by the expensive
vertices, and $\pos = \{1,2,\ldots,n\} \setminus \pos_\$$ to be the set of the remaining positions.
For  every $1 \leq i \leq |\pos|$, by $\posmap(i)$ we denote the $i$-th position of $\pos$.
Define also $\ord_\$: V_\$ \to \pos_\$$ as $\ord_\$(v) = p_v$.

We compute a set $\sol_\$$ consisting of all (unordered) pairs $v_1,v_2 \in V_\$$ such that
$v_1v_2 \notin E(G)$, but $p_{v_2} \in [p_{v_1}^L,p_{v_1}^R]$, that is,
the guessed values imply that $v_1v_2 \in E(G^\ord)$ and, consequently,
$\sol_\$ = \sol^\ord \cap \binom{V_\$}{2}$ in the correct branch.
Observe the following.
\begin{lemma}\label{lem:exp-sol}
In all branches
$\sol_\$$ is a completion of $G[V_\$]$, and $\ord_\$$, treated as an ordering
of $V_\$$, is an umbrella ordering of $G[V_\$] + \sol_\$$.
\end{lemma}
\begin{proof}
Consider any $a,b,c \in V_\$$ with $\ord_\$(a) < \ord_\$(b) < \ord_\$(c)$.
If $ac \in E(G) \cup \sol_\$$ then it follows from the clean-up operations and the definition of $\sol_\$$ that $\ord_\$(c) \in [p_a^L,p_a^R]$ and $\ord_\$(a) \in [p_c^L,p_c^R]$.
Recall that $\ord_\$(a) \in [p_a^L,p_a^R]$ and $\ord_\$(c) \in [p_c^L,p_c^R]$. 
Hence, $\ord_\$(b) \in [\ord_\$(a),\ord_\$(c)] \subseteq [p_a^L,p_a^R] \cap [p_c^L,p_c^R]$
and $ab,bc \in E(G) \cup \sol_\$$.
\end{proof}

Consider now a vertex $u \notin V_\$$.
For any $v \in V_\$$, if $uv \in E(G)$ then in the correct branch $\ord(u) \in [p_v^L,p_v^R]$.
This motivates us to define:
$$\pos_u = \posmap^{-1}\left(\pos \cap \bigcap_{v \in V_\$ \cap N_G(u)} [p_v^L,p_v^R]\right).$$
Observe that in the correct branch $\posmap^{-1}(\ord(u)) \in \pos_u$.

Furthermore, observe that, in the correct branch, if $uv \notin E(G)$ for some $u \notin V_\$$ and $v\in V_\$$, then exactly one of the following holds:
$uv \in \sol^\ord$ or $\ord(u) \notin [p_v^L,p_v^R]$. In other words, a vertex $v \in V_\$$ has degree exactly $p_v^R-p_v^L$ in
the graph $G^\ord$.
This motivates us to define the following cost value for every branch:
$$\cost_\$ = -|\sol_\$| + \sum_{v \in V_\$}( (p_v^R-p_v^L) - \deg_G(v)).$$
Observe that this cost function is actually meaningful for every branch:
\begin{lemma}\label{lem:exp-cost}
Let $\ord'$ be an  ordering  of $V(G)$ and $\sol$ be a completion of $G$ such that
\begin{itemize}
 \item[(i)]  $\ord'$ is an umbrella ordering
of $G+\sol$, and 
\item[(ii)]  for every $v \in V_\$$ we have
$\ord'(v) = p_v$ and $\ord'(N_{G+\sol}[v]) = [p_v^L,p_v^R]$.
Then there are exactly $\cost_\$$ edges of $\sol$ that are incident  with $V_\$$.
\end{itemize}
\end{lemma}
\begin{proof}
Observe that the degree of $v \in V_\$$ in $G+\sol$ is exactly $p_v^R-p_v^L$. Hence, exactly $p_v^R-p_v^L - \deg_G(v)$
edges of $\sol$ are incident with $v$
and the sum $\sum_{v \in V_\$} ((p_v^R-p_v^L) - \deg_G(v))$ counts the edges of $\sol$ incident with $V_\$$,
but double-counts the edges of $\sol$ with both endpoints in $V_\$$. However, 
the set of double-counted edges is exactly $\sol \cap \binom{V_\$}{2} = \sol_\$$. The lemma follows.
\end{proof}

We define graphs $\Gdown$ and $\Gup$ with vertex set $\{1,2,\ldots,|\pos|\}$ as follows.
For   $1 \leq i < j \leq |\pos|$, we set
$ij \in E(\Gdown)$ if and only if  there is a \emph{witness} vertex $x \in V_\$$ such that either  $p_x^L \leq \posmap(i) < \posmap(j) < p_x$, or  $p_x < \posmap(i) < \posmap(j) \leq p_x^R$. For $\Gup$, we set
$ij \notin E(\Gup)$ if and only if there exists a \emph{witness} vertex $y \in V_\$$ 
such that either 
$\posmap(i) < p_y^L \leq p_y < \posmap(j)$, or $\posmap(i) < p_y \leq p_y^R < \posmap(j)$.

The next lemma shows  that $\Gdown$ and $\Gup$ satisfy the requirements for being a part of a \spic{} instance.
\begin{lemma}\label{lem:exp-down-up}
Both $\Gdown$ and $\Gup$ are proper interval graphs and the identity is an umbrella ordering
of both of them. Moreover, in the correct branch
$E(\Gdown) \subseteq E(\posmap^{-1}((\ord(G^\ord))[\pos])) \subseteq E(\Gup)$.
\end{lemma}
\begin{proof}
For the first claim, observe that in the case of $\Gdown$, for every edge $ij \in E(\Gdown)$
with $i < j$,
its witness $x$ also witnesses that $i'j' \in E(\Gdown)$ for every $i \leq i' < j' \leq j$.
Similarly, in the case of $\Gup$, for any nonedge $ij \notin E(\Gup)$ with $i < j$,
its witness $y$ also witnesses that $i'j' \notin E(\Gup)$ for each $i' \leq i < j \leq j'$.

We now move to the second claim, so assume we are in the correct branch.
For $\Gdown$, observe that if $ij \in E(\Gdown)$,  then
$\ord^{-1}(\posmap(i))\ord^{-1}(\posmap(j)) \in E(G^\ord)$
by the umbrella property as $\ord^{-1}(p_x^L)\ord^{-1}(p_x) \in E(G^\ord)$
and $\ord^{-1}(p_x)\ord^{-1}(p_x^R) \in E(G^\ord)$.
For $\Gup$,  if $i,j$
are such that $\ord^{-1}(\posmap(i))\ord^{-1}(\posmap(j)) \in E(G^\ord)$
and $\posmap(i) < p_y < \posmap(j)$ for some $y \in V_\$$, then by the umbrella property we have that 
$y\ord^{-1}(\posmap(i)),y\ord^{-1}(\posmap(j)) \in E(G^\ord)$ and consequently
$p_y^L \leq \posmap(i) < p_y < \posmap(j) \leq p_y^R$. Since $y$ was chosen arbitrarily, it follows that $ij \in E(\Gup)$ and the lemma follows.
\end{proof}
By Lemma~\ref{lem:exp-down-up}, we may terminate the branches where $\Gdown$
is not a subgraph of $\Gup$.

Define $W = V(G) \setminus V_\$$, $H = G[W]$ and $\ell = k - \cost_\$$.
Recall that in the remaining branches $\mathcal{I} := (H,\ell,(\pos_u)_{u \in V(G)},\Gdown,\Gup)$ is a valid \spic{} instance.
In the next lemmata we show that it is sufficient to solve it instead of $(G,k)$.

\begin{lemma}\label{lem:pic-to-spic}
If $(G,k)$ is a YES-instance to \picname{} with the canonical umbrella ordering $\ord$,
then in the correct branch
the function $\ord_H := \posmap^{-1} \circ \ord|_W$
is a feasible ordering of the \spic{} instance $\mathcal{I}$
with $\sol^{\ord_H} \subseteq \sol^\ord \cap \binom{W}{2} =: \sol_W$;
in particular, for any $u \in W$ we have $|\incsol{\sol^{\ord_H}}{u}| \leq \tau$.
Moreover, $\cost(\ord_H) = |\sol^\ord| - \cost_\$ - |\sol_W \setminus \sol^{\ord_H}| \leq |\sol^\ord| - \cost_\$$.
\end{lemma}
\begin{proof}
Observe that $\ord_H$ is indeed  an ordering of $W$.
We first verify that it is feasible. Clearly, in the correct branch $\ord_H(u) = \posmap^{-1}(\ord(u)) \in \pos_u$ for any $u \in W$.
Consider any pair $u,v$ with 
$\ord_H(u) < \ord_H(v)$ and $\ord_H(u)\ord_H(v) \notin E(\Gup)$.
Let $y$ be a witness that $\ord_H(u)\ord_H(v) \notin E(\Gup)$.
If $\ord_H(u) < p_y^L \leq p_y < \ord_H(v)$ then
$uy \notin E(G^\ord)$ and, by the umbrella property, $uv \notin E(G^\ord)$, so in particular $uv\notin E(G)$.
Symmetrically, if $\ord_H(u) < p_y \leq p_y^R < \ord_H(v)$ then
$yv \notin E(G^\ord)$ and, by the umbrella property, $uv \notin E(G^\ord)$, so in particular $uv\notin E(G)$.
Consequently, $uv \notin E(G)$ in both cases and $\ord_H$ is feasible.

We now show that $\sol^{\ord_H} \subseteq \sol_W$.
Consider any $uv \in \sol^{\ord_H}$ and w.l.o.g. assume $\ord_H(u) < \ord_H(v)$.
If there exist $u',v' \in W$ with $\ord_H(u') \leq \ord_H(u) < \ord_H(v) \leq \ord_H(v')$
and $u'v' \in E(G)$, then $\ord(u') \leq \ord(u) < \ord(v) \leq \ord(v')$ by the monotonicity
of $\posmap$ and hence $uv \in \sol^\ord$.
Otherwise, by the definition of $\sol^{\ord_H}$, we have that $\ord_H(u)\ord_H(v) \in E(\Gdown)$.
By the definition of $\Gdown$, there exists $x \in V_\$$ with
$p_x^L \leq \posmap(\ord_H(u)) = \ord(u) < \ord(v) = \posmap(\ord_H(v)) < p_x$
or $p_x < \posmap(\ord_H(u)) = \ord(u) < \ord(v) = \posmap(\ord_H(v)) \leq p_x^R$.
In the first case, by the umbrella property we have that  $uv \in \sol^\ord$
because  $\ord^{-1}(p_x^L)x \in E(G^\ord)$.  Similarly, in the second case, $uv \in \sol^\ord$
  since $\ord^{-1}(p_x^R)x \in E(G^\ord)$.

We now compute the cost of $\ord_H$. By Lemma~\ref{lem:exp-cost}, there are exactly
$\cost_\$$ edges of $\sol^\ord$ incident with $V_\$$. Therefore
$|\sol_W| = |\sol^\ord| - \cost_\$$.
The already proven inclusion $\sol^{\ord_H} \subseteq \sol_W$ finishes the proof of the formula for the cost of $\ord_H$.
\end{proof}

\begin{lemma}\label{lem:spic-to-pic}
Let $\ord_H$ be a feasible ordering of the \spic{} instance $\mathcal{I}$ in some branch. Let also 
$\ord'$  be an ordering of $V(G)$ such that $\ord'(u) = \posmap(\ord_H(u))$ for $u \in W$
and $\ord'(u) = \ord_\$(u)$ for $u\not \in W$. Then $|\sol^{\ord'}|\leq \cost(\ord_H) + \cost_\$$.

\end{lemma}
\begin{proof}
We define
$$\sol = \sol^{\ord_H} \cup \sol_\$ \cup \{uv: u \in W \wedge v \in V_\$ \wedge uv \notin E(G) \wedge \posmap(\ord_H(u)) \in [p_v^L,p_v^R]\}.$$
We now show that $\ord'$ is an umbrella ordering of $G+\sol$.
Observe that if this is true, then Lemma~\ref{lem:exp-cost} will yield that $|\sol^{\ord'}| \leq |\sol| = |\sol^{\ord_H}| + c_\$$,
finishing the proof of the lemma; the condition (ii) of Lemma~\ref{lem:exp-cost} can be directly checked from the definitions of $\ord',\sol$.

Consider then a triple $a,b,c \in V(G)$ with $\ord'(a) < \ord'(b) < \ord'(c)$
and $ac \in E(G) \cup \sol$. We consider a few cases, depending on the
intersection $V_\$ \cap \{a,b,c\}$.

First, consider the case $a,c \in V_\$$.
If $b \in V_\$$, then $ab,bc \in E(G) \cup \sol$ by Lemma~\ref{lem:exp-sol}.
Otherwise, observe that the cleanup operation imply that $\ord'(a)=p_a\in [p_c^L,p_c^R]$ and $\ord'(c)=p_c\in [p_a^L,p_a^R]$ and we obtain
$\ord'(b) = \posmap(\ord_H(b)) \in [p_a^L,p_a^R] \cap [p_c^L,p_c^R]$.
Hence $ab,bc \in E(G) \cup \sol$ directly from the definition of $\sol$.

Second, consider the case $a \in V_\$$ and $c \in W$.
We claim that $ac \in E(G) \cup \sol$ implies that $\ord'(c) = \posmap(\ord_H(c)) \in [p_a^L,p_a^R]$. Indeed, if $ac\in \sol$ then this follows directly from the definition of $\sol$. If $ac\in E(G)$, however, then $\ord'(c) = \posmap(\ord_H(c))\in \posmap(\pos_c)\subseteq [p_a^L,p_a^R]$ since $\ord_H$ is feasible.
Now observe that since $\ord'(a) = p_a \in [p_a^L,p_a^R]$, then we have also that $\ord'(b) \in [p_a^L,p_a^R]$. Since $\ord'(a) < \ord'(b) < \ord'(c)$, then in fact $\ord'(b),\ord'(c)\in [p_a,p_a^R]$.

Assume first that $b \in V_\$$.
Then $ab \in E(G) \cup \sol_\$$ by the definition of $\sol_\$$.
Moreover, as $\ord'(b) = p_b > \ord'(a) = p_a$, by the cleanup operations
we have that $p_b^R \geq p_a^R$ and, consequently, $\ord'(c) = \posmap(\ord_H(c)) \in [p_b^L,p_b^R]$.
Hence, in this case $bc \in E(G) \cup \sol$ by the definition of $\sol$.

Assume now $b \in W$. Clearly $\ord'(b) \in [p_a^L,p_a^R]$ implies that $ab \in E(G) \cup \sol$ by the definition of $\sol$.
Moreover, observe that as both $\ord'(b) = \posmap(\ord_H(b))$ and $\ord'(c) = \posmap(\ord_H(c))$
belong to $[p_a,p_a^R]$, we have $\ord_H(b)\ord_H(c) \in \Gdown$ and hence $bc \in E(G) \cup \sol^{\ord_H}$.

Third, observe that the case $a \in W$ and $c \in V_\$$ is symmetrical to the previous one.

Finally, consider the case $a,c \in W$, so $ac \in E(G) \cup \sol^{\ord_H}$.
If $b \in W$ then $ab,bc \in E(G) \cup \sol^{\ord_H}$
as $\ord_H$ is an umbrella ordering of $G[W]+\sol^{\ord_H}$. Hence, assume $b \in V_\$$.
Observe that $ac \in E(G) \cup \sol^{\ord_H}$ implies that $ac \in E(\Gup)$.
However, we have that $\posmap(\ord_H(a)) < p_b < \posmap(\ord_H(c))$.
Thus, by the definition of $\Gup$, we have $p_b^L \leq \posmap(\ord_H(a)) < \posmap(\ord_H(c)) \leq p_b^R$ and, by the definition of $\sol$, $ab,bc \in E(G) \cup \sol$.
This concludes the proof of the lemma.
\end{proof}
\begin{lemma}\label{lem:pic-spic-equiv}
If $(G,k)$ is a YES-instance to \picname{} with the canonical umbrella ordering $\ord$,
then in the correct branch
the function $\ord_H := \posmap^{-1} \circ \ord|_W$
is the canonical umbrella ordering of the \spic{} instance $\mathcal{I}$
of cost at most $\ell$.
Moreover, $\sol^{\ord_H} = \sol^\ord \cap \binom{W}{2}$; in particular,
for any $u \in W$ we have $|\incsol{\sol^{\ord_H}}{u}| \leq \tau$.
\end{lemma}
\begin{proof}
We focus on the correct branch.
By Lemma~\ref{lem:pic-to-spic}, there exists a feasible ordering
of the \spic{} instance $\mathcal{I}$.
Let $\ord_H'$ be the canonical ordering of this instance.
Define $\ord'$ as in Lemma~\ref{lem:spic-to-pic} for the ordering $\ord_H'$.

By Lemma~\ref{lem:spic-to-pic} and the optimality of $\ord$,
we have that
$$|\sol^\ord| \leq |\sol^{\ord'}| \leq \cost(\ord_H') + \cost_\$.$$
On the other hand, by Lemma~\ref{lem:pic-to-spic} and the optimality of $\ord_H'$, 
we have that
$$\cost(\ord_H') \leq \cost(\posmap^{-1} \circ \ord|_W) \leq |\sol^\ord| - \cost_\$.$$
Hence, all aforementioned inequalities are in fact equalities, and $\sol^{\ord_H} = \sol^\ord \cap \binom{W}{2}$.
In particular,
$\sol^{\ord'}$ is a minimum completion of $G$ and $\posmap^{-1} \circ \ord|_W$
is of minimum possible cost.
By the monotonicity of $\posmap$, we infer that the lexicographical minimization
in fact chooses $\ord_H' = \ord_H$ and the lemma is proven.
\end{proof}

In the next sections we will show the following.
\begin{theorem}\label{thm:spic}
There exists an algorithm that, given a branch with a \spic{} instance $\mathcal{I}$,
runs in time $n^{\Oh(\ell/\tau + \tau^2)}$
and, if given the correct branch, computes the canonical ordering of $\mathcal{I}$.
\end{theorem}
The equivalence shown in Lemmata~\ref{lem:pic-to-spic}, \ref{lem:spic-to-pic}
and~\ref{lem:pic-spic-equiv}, together with the bound $n = \Oh(k^3)$, allows us to solve the \picname{} instance
$(G,k)$ by applying the algorithm of Theorem~\ref{thm:spic} to each branch separately. Observe that we have $k^{\Oh(k^{2/3})}$ branches, and for $\tau=(2k)^{1/3}$, $\ell\leq k$ and $n=\Oh(k^3)$ we have $n^{\Oh(\ell/\tau + \tau^2)}=k^{\Oh(k^{2/3})}$; therefore, the running time will be as guaranteed in Theorem~\ref{thm:main}.

Hence, it remains to prove Theorem~\ref{thm:spic}.
In its proof it will be clear that the algorithm runs within the given time bound.
Hence, we assume that we work in the correct branch and we will mostly focus on proving 
that we indeed find the canonical ordering of $\mathcal{I}$.

%

\section{Sections}\label{sec:sections}
\providecommand{\secfam}{\mathcal{S}}
\providecommand{\ptcfam}{\mathcal{T}}
\providecommand{\ptc}{\Lambda}

We now proceed with the proof of Theorem~\ref{thm:spic}.
Assume we are given the correct branch with a \spic{} instance $\mathcal{I} = (H,\ell,(\pos_u)_{u \in V(G)},\Gdown,\Gup)$.
Recall that we look for the canonical ordering $\ord_H$ of $\mathcal{I}$
and we assume that $\ord_H$ is of cost at most $\ell$ and $|\incsol{\sol^{\ord_H}}{u}| \leq \tau$ for every $u \in V(G)$.
The last assumption 
allows us to guess edges $\incsol{\sol^{\ord_H}}{u}$ for a set  of carefully chosen vertices
$u \in V(H)$. 
In this section we use this property to show the following statement.

\begin{definition}
A \emph{section} is a subset $A$ of $V(H)$. 
A section $A$ is \emph{consistent} with an  ordering $\ord_H$ 
if $\ord_H$ maps $A$ onto the first $|A|$ positions.
\end{definition}

\begin{theorem}\label{thm:sec-enum}
In $k^{\Oh(\tau)}$ time one can enumerate a family
$\secfam$ of $k^{\Oh(\tau)}$ sections
that contains all sections consistent with the canonical ordering $\ord_H$.
\end{theorem}

The proof of Theorem~\ref{thm:sec-enum} is divided into two steps.
First, we investigate true twin classes in the graph $H^{\ord_H}$,
and show that we can efficiently enumerate a small family
of candidates for these twin classes. Then we use  the twin class
residing at position $|A|+1$ to efficiently `guess' a
section $A$ consistent with the canonical ordering $\ord_H$.
Henceforth we assume that the canonical ordering $\ord_H$ is of cost at most $k$.

\subsection{Potential twin classes}

Recall that two vertices $x$ and $y$ are \emph{true twins} if $N[x] = N[y]$;
in particular, this implies that they are adjacent. The relation
of being a true twin is an equivalence relation, and an equivalence class
of this relation is called a \emph{twin class}.
We remark the following observation, straightforward from the definition of
an umbrella ordering.
\begin{lemma}\label{lem:twins-together}
In an umbrella ordering of a proper interval graph, the vertices of any twin
class occupy consecutive positions.
\end{lemma}
The main result of this section is the following.
\begin{theorem}\label{thm:ptc-enum}
In $k^{\Oh(\tau)}$ time one can enumerate a family
$\ptcfam$ of $k^{\Oh(\tau)}$ triples $(L,\ptc,\ord_\ptc)$
such that for any twin class $\ptc$ of $H^{\ord_H}$,
if $L$ is the set of vertices of $H$ placed to the left of $\ptc$ in the ordering $\ord_H$,
then $(L,\ptc,\ord_H|_\ptc) \in \ptcfam$.
\end{theorem}
We describe the algorithm of Theorem~\ref{thm:ptc-enum} as a branching algorithm
that produces $k^{\Oh(\tau)}$ subcases and, in each subcase, produces 
one pair $(L,\ptc,\ord_\ptc)$. We fix one twin class $\ptc$ of $H^{\ord_H}$ and argue
that the algorithm in one of the branches produces $(L,\ptc,\ord_H|_\ptc)$, where
$L$ is defined as in Theorem~\ref{thm:ptc-enum}.
We perform this task in two phases: we first reason about $L$ and $\ptc$, and then
we deduce the ordering $\ord_H|_\ptc$.

\subsubsection{Phase one: $L$ and $\ptc$}

The algorithm guesses the following five vertices (see also Figure~\ref{fig:lambda}):
\begin{enumerate}
\item $a$ is any vertex of $\ptc$,
\item $b_1$ is the rightmost vertex outside $N_{H^{\ord_H}}[\ptc]$ in $\ord_H$ that lies before $\ptc$, or $b_1 = \bot$ if no such vertex exists;
\item $c_1$ is the leftmost vertex of $N_{H^{\ord_H}}[\ptc]$ in $\ord_H$;
\item $c_2$ is the rightmost vertex of $N_{H^{\ord_H}}[\ptc]$ in $\ord_H$;
\item $b_2$ is the leftmost vertex outside $N_{H^{\ord_H}}[\ptc]$ in $\ord_H$ that lies after $\ptc$, or $b_2 = \bot$ if no such vertex exists.
\end{enumerate}

\begin{figure}
\centering
\includegraphics{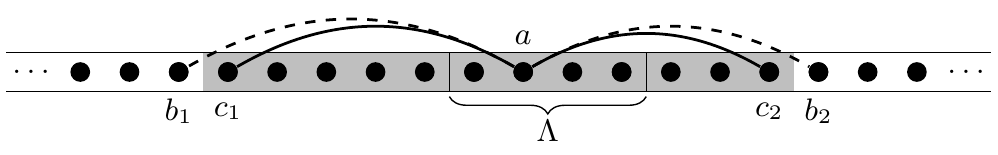}
\caption{The guessed vertices $a$, $b_1$, $b_2$, $c_1$ and $c_2$ with respect
  to a twin class $\ptc$. The gray area denotes $N_{H^{\ord_H}}(\ptc)$.}
\label{fig:lambda}
\end{figure}

Moreover, for each $u \in \{a,b_1,b_2,c_1,c_2\} \setminus \{\bot\}$ the algorithm guesses
$\incsol{\sol^{\ord_H}}{u}$. This leads us to $k^{\Oh(\tau)}$ subcases.
We now argue that, if the guesses are correct, we can deduce the pair $(L,\ptc)$.
The crucial step is the following.
\begin{lemma}\label{lem:ptc-crux}
In the branch where the guesses are correct, the following holds for any $u \in N_{H^{\ord_H}}[a]$,
\begin{enumerate}
\item if $u \in N_{H^{\ord_H}}[b_1]$ or $u \notin N_{H^{\ord_H}}[c_2]$, then $u \notin \ptc$ and $u$ lies before $\ptc$ in the ordering $\ord_H$;
\item if $u \in N_{H^{\ord_H}}[b_2]$ or $u \notin N_{H^{\ord_H}}[c_1]$, then $u \notin \ptc$ and $u$ lies after $\ptc$ in the ordering $\ord_H$;
\item if none of the above happens, then $u \in \ptc$.
\end{enumerate}
Here we take the convention that $N_{H^{\ord_H}}[\bot] = \emptyset$.
\end{lemma}
\begin{proof}
By the definition of $b_1$, $b_2$, $c_1$ and $c_2$, we have that every vertex $u \in \ptc$ lies in $N_{H^{\ord_H}}[c_1]$ and $N_{H^{\ord_H}}[c_2]$,
but not in $N_{H^{\ord_H}}[b_1]$ nor in $N_{H^{\ord_H}}[b_2]$. Consequently, any vertex of $\ptc$ falls into the third category of the statement of the lemma.

We now show that any other vertex of $N_{H^{\ord_H}}[a]$ falls into one of the first two categories, depending on its position in the ordering $\ord_H$.
By symmetry, we may only consider a vertex $u \in N_{H^{\ord_H}}[a] \setminus \ptc$ that lies before $\ptc$ in $\ord_H$.
Note that the umbrella property together with $a \notin N_{H^{\ord_H}}[b_2]$ implies that $u \notin N_{H^{\ord_H}}[b_2]$,
and together with $ac_1 \in E(H^{\ord_H})$ implies $uc_1 \in E(H^{\ord_H})$.
Consequently, $u$ does not fall into the second category in the statement of the lemma. We now show that it falls into the first one.

As $u \notin \ptc$ and $u\in N_{H^{\ord_H}}[a]$, either $N_{H^{\ord_H}}(u) \setminus N_{H^{\ord_H}}[a]$ is not empty or $N_{H^{\ord_H}}(a) \setminus N_{H^{\ord_H}}[u]$ is not empty.
In the first case, let $uw \in E(H^{\ord_H})$ but $aw \notin E(H^{\ord_H})$. Since also $ua\in E(H^{\ord_H})$, by the umbrella property it easily follows that $w$ lies before $u$ in the ordering $\ord_H$, so in particular before $\ptc$.
By the definition of $b_1$, $b_1$ exists and $\ord_H(b_1) \geq \ord_H(w)$. By the umbrella property, $b_1u \in E(H^{\ord_H})$ and hence $u \in N_{H^{\ord_H}}[b_1]$.

In the second case, assume $uw \notin E(H^{\ord_H})$ but $aw \in E(H^{\ord_H})$. Again, since $ua\in E(H^{\ord_H})$, by the umbrella property it easily follows that $w$ lies after $\ptc$ in the ordering $\ord_H$, so in particular after $u$.
By the definition of $c_2$ and the existence of $w$, $c_2 \notin \ptc$ and $\ord_H(c_2) \geq \ord_H(w)$. By the umbrella property, $c_2u \notin E(H^{\ord_H})$ and $u \notin N_{H^{\ord_H}}[c_2]$.
Hence, $u$ falls into the first category and the lemma is proven.
\end{proof}

The knowledge of $a$ and $\incsol{\sol^{\ord_H}}{a}$ allows us to compute
$N_{H^{\ord_H}}[\ptc] = N_{H^{\ord_H}}[a]$. By making use of Lemma~\ref{lem:ptc-crux}, we can further partition
$N_{H^{\ord_H}}[\ptc]$ into $\ptc$, the vertices of $N_{H^{\ord_H}}(\ptc)$
that lie before $\ptc$ in the ordering $\ord_H$, and the ones that lie after $\ptc$.
We are left with the vertices outside $N_{H^{\ord_H}}[\ptc]$.

We  guess the position $i$ such that the first vertex of $\ptc$ in the ordering $\ord_H$
is in position $i$.
Note that, by Lemma~\ref{lem:twins-together}, the vertices of $\ptc$
occupy positions $i,i+1,\ldots,i+|\ptc|-1$ in $\ord_H$.

Let $C$ be a connected component of $H \setminus N_{H^{\ord_H}}[\ptc]$.
Recall that by Lemma~\ref{lem:pic-spic-equiv}, $\ord_H = \posmap^{-1} \circ \ord|_W$
and $\sol^{\ord_H} = \sol^\ord \cap \binom{W}{2}$.
As no vertex of $C$ is incident with $\ptc$ in $H^{\ord_H}$, by the properties
of an umbrella ordering we infer that all vertices of $N_G[C]$ lie before position $\posmap(i)$
or all vertices of $N_G[C]$ lie after position $\posmap(i+|\ptc|-1)$ in the ordering $\posmap \circ \ord_H = \ord|_W$.
As $G$ is assumed to be connected, $N_G(C)$
contains a vertex of $N_{H^{\ord_H}}[\ptc]$ or of $V_\$$.
Any such vertex allows us to deduce which of the two aforementioned options is true for $C$ in $\ord$.
This allows us to decide whether $C \subseteq L$ or $L \cap C = \emptyset$, and consequently deduce the set $L$. Note that it must hold that $|L|=i-1$, and otherwise we may discard the guess.

\subsubsection{Phase two: the ordering $\ord_H|_\ptc$}\label{sss:ptc-phase-2}

We are left with determining $\ord_H|_\ptc$.
Note that we already know the domain $\ptc$ and the codomain $\{i,i+1,\ldots,i+|\ptc|-1\}$
of this bijection.
We prove the following.
\begin{lemma}\label{lem:ptc-replacement}
The bijection $\ord_H|_\ptc$ is 
the lexicographically minimum bijection
$\ord_\ptc:\ptc \to \{i,i+1,\ldots,i+|\ptc|-1\}$
among those bijections $\ord_\ptc$ that satisfy $\ord_\ptc(u) \in \pos_u$ for any $u \in \ptc$.
\end{lemma}
\begin{proof}
Let $\ord_\ptc:\ptc \to \{i,i+1,\ldots,i+|\ptc|-1\}$ be the lexicographically minimum bijection
among those that satisfy $\ord_\ptc(u) \in \pos_u$ for any $u \in \ptc$; note that at least one such bijection exists, since $\ord_H|_\ptc$ is one.
Consider an ordering $\ord'$ of $V(H)$
defined as follows: $\ord'(u) = \ord_\ptc(u)$ if $u \in \ptc$ and $\ord'(u) = \ord_H(u)$ otherwise.
Observe that $\ord'$ is an ordering of $V(H)$.
Moreover, as $\ptc$ is a twin class of $H^{\ord_H}$, we have
$\ord'(H^{\ord_H}) = \ord(H^{\ord_H})$. Hence $\ord'$ is a feasible ordering of $H$ and umbrella ordering of $H^{\ord_H}$.
We infer that $\sol^{\ord'} \subseteq \sol^{\ord_H}$.
On the other hand, as $\ord_H$ is the canonical solution, we have $\cost(\ord_H) \leq \cost(\ord')$.
Hence, both aforementioned inequalities are in fact tight and $\sol^{\ord'} = \sol^{\ord_H}$.
Furthermore, the lexicographical minimization criterion
implies that $\ord_\ptc = \ord_H|_\ptc$ and $\ord' = \ord_H$.
\end{proof}

Finally, observe that the characterization of $\ord_H|_\ptc$ given by Lemma~\ref{lem:ptc-replacement}
fits into the conditions of Lemma~\ref{lem:lex-match} and, consequently, $\ord_H|_\ptc$
can be computed in polynomial time given $L$, $\ptc$ and the index $i$.
This concludes the proof of Theorem~\ref{thm:ptc-enum}.

\subsection{Proof of Theorem~\ref{thm:sec-enum}}

Given Theorem~\ref{thm:ptc-enum}, the proof of Theorem~\ref{thm:sec-enum} is now straightforward.
We first compute the family $\ptcfam$ of Theorem~\ref{thm:ptc-enum}.
Then, for each $(L,\ptc,\ord_\ptc) \in \ptcfam$ and each position $p\in \{1,2,\ldots,|V(H)|\}$ we output a set
$$A := L \cup \{u \in \ptc: \ord_\ptc(u) < p\}.$$
Additionally, we output a section $V(H)$.
Clearly, the algorithm outputs $k^{\Oh(\tau)}$ sections and works within the promised
time bound.
It remains to argue that it outputs all sections consistent with $\ord_H$.

Consider a section $A$ consistent with $\ord_H$, that is, $A = \ord^{-1}(\{1,2,\ldots,|A|\})$.
If $A = V(H)$, the statement is obvious, so assume otherwise.
Consider the position $p := |A|+1$, let $u = \ord^{-1}(p)$
and let $\ptc$ be the twin class of $u$ in $H^{\ord_H}$.
Moreover, let $L$ be the set of vertices of $H$ placed before $\ptc$
in $\ord_H$. By Theorem~\ref{thm:ptc-enum}, $(L,\ptc,\ord_H|_\ptc) \in \ptcfam$.
Moreover, note that the algorithm outputs exactly the set $A$ when it considers the triple
$(L,\ptc,\ord_H|_\ptc)$ and position $p$.
This concludes the proof of Theorem~\ref{thm:sec-enum}.

\section{Dynamic programming}\label{sec:dp}
\providecommand{\jump}[1]{\mathtt{jump}(#1)}
\providecommand{\jumpset}[1]{X_{#1}}
\providecommand{\secp}[1]{A_{#1}}
\providecommand{\jumpfam}{\mathcal{J}}
\providecommand{\chainfam}{\mathcal{C}}

In this section we conclude the proof of Theorem~\ref{thm:spic}
by showing the following.
\begin{theorem}\label{thm:dp}
Given a \spic{} instance $\mathcal{I} = (G,k,(\pos_u)_{u \in V(G)},\Gdown,\Gup)$ with $n = |V(G)|$,
 a threshold $\tau$ and a family $\secfam \subseteq 2^{V(G)}$,
one can in $n^{\Oh(k/\tau+\tau)} |\secfam|^{\Oh(\tau)}$ time find the canonical ordering $\ord$ of $\mathcal{I}$, assuming that
\begin{enumerate}
\item $\cost(\ord) \leq k$;
\item for each $u \in V(G)$, $|\incsol{\sol^\ord}{u}| \leq \tau$;
\item each section consistent with $\ord$ belongs to $\secfam$.
\end{enumerate}
\end{theorem}
Observe that if we apply Theorem~\ref{thm:dp} to a branch with a \spic{} instance $\mathcal{I}$,
the threshold $\tau$ and family $\secfam$ output by Theorem~\ref{thm:sec-enum},
then we obtain the algorithm promised by Theorem~\ref{thm:spic}.

The algorithm of Theorem~\ref{thm:dp} is a dynamic programming algorithm.
Henceforth assume that the instance $\mathcal{I}$ with threshold $\tau$ and family $\secfam$
is as promised in the statement of Theorem~\ref{thm:dp}, and let $\ord$ be the canonical ordering of $\mathcal{I}$.
We develop two different ways of separating the graphs $G$ and $G^\ord$ into smaller parts,
suitable for dynamic programming. Consequently, the dynamic programming algorithm
has in some sense `two layers', and two different types of states.

\subsection{Layer one: jumps and jump sets}

We first develop a way to split the graphs $G$ and $G^\ord$ `vertically'.
To this end, first denote for any position $p$
the section $\secp{p} = \{v \in V(G): \ord(v) < p\}$; note that this definition also makes sense
for $p = \infty$ and $\secp{\infty} = V(G)$.
Second, for any position $p$ define
$$\jump{p} = \min\{q: q > p \wedge \ord^{-1}(p)\ord^{-1}(q) \notin E(G^\ord)\};$$
in this definition we follow the convention that the minimum of an empty set is $\infty$. Moreover, we define a \emph{jump set} for position $p$ as
$$\jumpset{p} = \ord^{-1}([p,\jump{p}-1]) = \secp{\jump{p}} \setminus \secp{p}.$$
See also Fig.~\ref{fig:jump} for an illustration.

\begin{figure}
\centering
\includegraphics{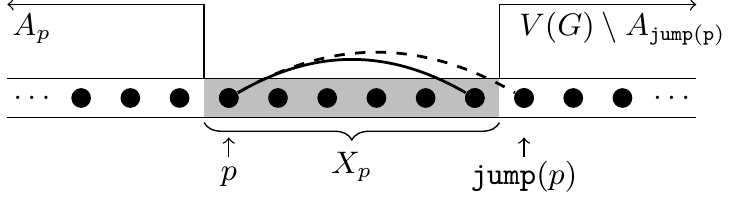}
\caption{A jump at position $p$ and the corresponding jump set.
  The jump set $\jumpset{p}$, denoted with gray, is a clique in $G^\ord$, and no edge of $G^\ord$
    connects $\secp{p}$ with $V(G) \setminus \secp{\jump{p}}$.}
\label{fig:jump}
\end{figure}

The next two lemmata follow directly from the definition of a jump and the properties
of  umbrella orderings.
\begin{lemma}\label{lem:jump-ineq}
For any positions $p$ and $q$, if $p \leq q$ then $\jump{p} \leq \jump{q}$.
\end{lemma}
\begin{lemma}\label{lem:jump-cut} Jump set 
$\jumpset{p}$ is  a clique in $G^\ord$, but no edge of $G^\ord$
connects a vertex of $\secp{p}$ with a vertex of $V(G) \setminus \secp{\jump{p}}$.
\end{lemma}

We now slightly augment the graph $G$ so that $\jump{p} \neq \infty$ for all interesting
positions; see also Figure~\ref{fig:augment}.
We take $\Oh(n^2)$ branches, guessing the first
and the last vertex of $G$ in the ordering $\ord$; denote them by $\alpha$ and $\omega$.
We introduce  new vertices, $\alpha_1,\alpha_2,\omega_1,\omega_2,\omega_3$ 
and new edges $\alpha_1\alpha_2,\alpha\alpha_1,\omega_2\omega_3,\omega_1\omega_2,\omega\omega_1$ in $G$.
We also introduce new positions $-1,0,n+1,n+2,n+3$, isolated in $\Gdown$ and connected
by edges $\{-1,0\},\{0,1\},\{n,n+1\},\{n+1,n+2\},\{n+2,n+3\}$ in $\Gup$.
We define $\pos_{\alpha_1} = \{0\}$, $\pos_{\alpha_2} = \{-1\}$, $\pos_{\omega_1} = \{n+1\}$,
   $\pos_{\omega_2} = \{n+2\}$ and $\pos_{\omega_3} = \{n+3\}$.
Moreover, we put $\alpha_2$ and $\alpha_1$ before all vertices of $G$ in the ordering $\ord_0$,
and $\omega_1$, $\omega_2$ and $\omega_3$ after them.
Note that, if we precede all the vertices in the ordering $\ord$ with
$\alpha_2,\alpha_1$ and succeed with $\omega_1,\omega_2,\omega_3$ we obtain an ordering with no
higher cost. Due to the way we have extended $\ord_0$ to the new vertices, the
extended ordering $\ord$ defined in this way is the canonical ordering
of the extended graph $G$. Hence, we may abuse the notation and denote by $G$
the graph after the addition of these five new vertices, and assume that $V(\Gdown) = V(\Gup) = \{1,2,\ldots,|V(G)|\}$ again.

\begin{figure}
\centering
\includegraphics{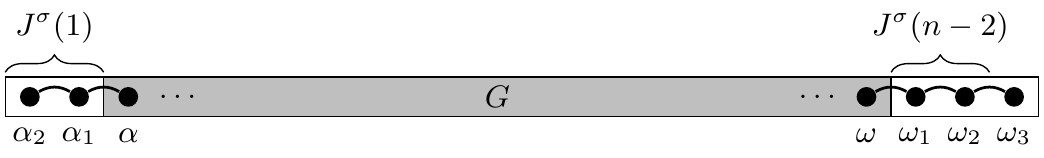}
\caption{Augmentation of the input graph $G$.}
\label{fig:augment}
\end{figure}

Observe now that $\jump{1} = 3$ and $\jumpset{1} = \{\alpha_2,\alpha_1\}$,
 as $\ord^{-1}(1) = \alpha_2$ and $\ord^{-1}(3) = \alpha$.
Moreover, $\jump{n-2} = n$ and $\jumpset{n-2} = \{\omega_1,\omega_2\}$,
as $\ord^{-1}(n-2) = \omega_1$, $\ord^{-1}(n-1) = \omega_2$, and $\ord^{-1}(n) = \omega_3$

The main observation now is that a jump set, together with all edges of $\sol^\ord$
incident with it (i.e., $\incsol{\sol^\ord}{\jumpset{p}}$) contains
all sufficient information to divide the problem into parts before and after
a jump set.
\begin{lemma}\label{lem:jump-equiv}
For any position $p$, the following holds.
\begin{enumerate}
\item
For any $u_1,u_2 \in \jumpset{p}$ such that $\ord(u_1) \leq \ord(u_2)$ we have
\begin{align}
N_{G^\ord}(u_1) \cap \secp{p} &\supseteq N_{G^\ord}(u_2) \cap \secp{p}, \label{eq:jump-left} \\
N_{G^\ord}(u_1) \setminus \secp{\jump{p}} &\subseteq N_{G^\ord}(u_2) \setminus \secp{\jump{p}}. \label{eq:jump-right}
\end{align}
\item For any bijection $\ord_p: \jumpset{p} \to [p,\jump{p}-1]$
such that $\ord_p(u) \in \pos_u$ for any $u \in \jumpset{p}$ and 
both inclusions~\eqref{eq:jump-left} and~\eqref{eq:jump-right} hold
for any $u_1,u_2 \in \jumpset{p}$ with $\ord_p(u_1)  \leq \ord_p(u_2)$,
if we define an ordering $\ord'$ of $V(G)$ as
$\ord'(u) = \ord_p(u)$ if $u \in \jumpset{p}$ and $\ord'(u) = \ord(u)$ otherwise,
then $\ord'$ is feasible and $\ord'(G^{\ord'})$ is a subgraph of $\ord(G^\ord)$. 
\end{enumerate}
\end{lemma}
\begin{proof}
The first statement is straightforward from the properties of an umbrella ordering.
Let $\ord_p$ and $\ord'$ be as in the second statement.
Observe that inclusions~\eqref{eq:jump-left} and~\eqref{eq:jump-right}, together with the fact that $X_p$ is a clique in $\ord(G^\ord)$,
imply that $\ord'$ and $\ord$ differ only on the internal order of twin classes of $G^\ord$ and consequently $\ord'(G^\ord) = \ord(G^\ord)$.
Together with the fact that $\ord'(u)\in \pos_u$ for any $u\in V(G)$, this means that $\ord'$ is a feasible ordering of $G$ and an umbrella ordering of $G^\ord$.
Consequently $\sol^{\ord'} \subseteq \sol^\ord$, $\ord'(G^{\ord'})$ is a subgraph of $\ord'(G^\ord) = \ord(G^\ord)$, and the lemma is proven.
\end{proof}

We use Lemma~\ref{lem:jump-equiv} to fit the task of computing $\ord|_{\jumpset{p}}$ into Lemma~\ref{lem:lex-match}.
\begin{lemma}\label{lem:jump-detect}
Given a position $p$ and the sets $\jumpset{p}$, $\secp{p}$ and
$\incsol{\sol^\ord}{\jumpset{p}}$, one can
in polynomial time compute the ordering $\ord|_{\jumpset{p}}$.
\end{lemma}
\begin{proof}
First, observe that the data promised in the lemma statement allows
us to compute $N_{G^\ord}(u) \cap \secp{p}$ and $N_{G^\ord}(u) \setminus \secp{\jump{p}}$
for every $u \in \jumpset{p}$.
Define a binary relation $\preceq$ on $\jumpset{p}$ as $u_1 \preceq u_2$ if and only if
both~\eqref{eq:jump-left} and~\eqref{eq:jump-right} hold for $u_1$ and $u_2$.
Lemma~\ref{lem:jump-equiv} asserts that $\preceq$ is a total quasi-order on $\jumpset{p}$.
That is, the set $\jumpset{p}$ can by partitioned into sets $U_1,U_2,\ldots,U_s$
such that $u_1 \preceq u_2$ and $u_2 \preceq u_1$ for any $1 \leq j \leq s$ and
$u_1,u_2 \in U_j$, and $u_1 \preceq u_2$, $u_2 \not\preceq u_1$ for any $1 \leq j_1 < j_2 \leq s$
and $u_1 \in U_{j_1}$, $u_2 \in U_{j_2}$.
(Formally, we terminate the current branch if $\preceq$ does not satisfy these properties.)

Observe that $\ord|_{\jumpset{p}}$ maps $\jumpset{p}$ onto $[p,\jump{p}-1]$.
Lemma~\ref{lem:jump-equiv} asserts that all vertices of $U_1$
are placed by $\ord$ on the first $|U_1|$ positions of the range of $\ord|_{\jumpset{p}}$,
all vertices of $U_2$ are placed on the next $|U_2|$ positions etc.
We use Lemma~\ref{lem:lex-match} to find a lexicographically minimum ordering $\ord_p$ that satisfies the above
and additionally $\ord_p(u) \in \pos_u$ for each $u \in \jumpset{p}$.
Define $\ord'$ as in Lemma~\ref{lem:jump-equiv}.
By the minimality of $\ord$, we have $\cost(\ord') \geq \cost(\ord)$, but Lemma~\ref{lem:jump-equiv}
asserts that $\ord'(G^{\ord'})$ is a subgraph of $\ord(G^\ord)$.
Hence, $\ord'$ is of minimum possible cost.
By the lexicographical minimality of $\ord_p$, we have $\ord_p = \ord|_{\jumpset{p}}$
and the lemma is proven.
\end{proof}


With help of family $\secfam$, Lemma~\ref{lem:jump-detect} allows us to efficiently enumerate
jump sets with their surroundings.
\begin{theorem}\label{thm:jump-enum}
One can in $n^{\Oh(k/\tau)}|\secfam|^2$ time enumerate a family $\jumpfam$
of at most $n^{\Oh(k/\tau)}|\secfam|^2$ tuples $(A,X,\ord_X)$ such that:
\begin{enumerate}
\item in each tuple $(A,X,\ord_X)$ we have
\begin{enumerate}
\item $A,X \subseteq V(G)$ and $A \cap X = \emptyset$,
\item $\Gdown[[|A|+1,|A|+|X|]]$ is a complete graph,
\item $\ord_X$ is a bijection between $X$ and $[|A|+1,|A|+|X|]$;
\end{enumerate}
\item for any position $p$, if there are at most $2k/\tau$
edges of $\sol^\ord$ incident to $\jumpset{p}$, then
the tuple $J^\ord(p) := (\secp{p},\jumpset{p},\ord|_{\jumpset{p}})$ belongs to $\jumpfam$.
\end{enumerate}
\end{theorem}
\begin{proof}
We provide a procedure of guessing at most $n^{\Oh(k/\tau)}|\secfam|^2$ candidate tuples that will constitute the family $\jumpfam$. Since the promised properties of elements of $\jumpfam$ can be checked in polynomial time, it suffices to argue that every triple of the form $(\secp{p},\jumpset{p},\ord|_{\jumpset{p}})$ will be among the guessed candidates.

The number of choices for $\secp{p}$ and $\secp{\jump{p}}$ is $|\secfam|^2$.
Observe that then $\jumpset{p} = \secp{\jump{p}} \setminus \secp{p}$.
Furthermore, there are $n^{\Oh(k/\tau)}$ ways to choose $\incsol{\sol^\ord}{\jumpset{p}}$
and, by Lemma~\ref{lem:jump-detect}, we can further deduce $\ord|_{\jumpset{p}}$. Finally,  observe that  by the definition of a jump it follows that every triple $(\secp{p},\jumpset{p},\ord|_{\jumpset{p}})$ satisfies the promised properties of the elements of $\jumpfam$.
\end{proof}

We are now ready to describe the first layer of our dynamic programming algorithm.
\begin{definition}[layer-one state]
A \emph{layer-one state} is a pair $(J^1,J^2)$ of two elements of $\jumpfam$,
$J^1 = (A^1,X^1,\ord_X^1)$, $J^2 = (A^2,X^2,\ord_X^2)$ such that 
$A^1 \subseteq A^2$ and $(A^1 \cup X^1) \subseteq (A^2 \cup X^2)$.
The \emph{value} of a layer-one state $(J^1,J^2)$ is a bijection
$f[J^1,J^2] : (A^2 \cup X^2) \setminus A^1 \to [|A^1|+1,|A^2\cup X^2|]$ satisfying the following:
\begin{enumerate}
\item $f[J^1,J^2]$ is a feasible ordering of its domain, that is,
  for any $u \in (A^2 \cup X^2) \setminus A^1$ we have $f[J^1,J^2](u) \in \pos_u$ and
  for any $u_1,u_2 \in (A^2 \cup X^2) \setminus A^1$ such that $u_1u_2\in E(G)$, we have
  $f[J^1,J^2](u_1)f[J^1,J^2](u_2) \in E(\Gup)$;
\item $f[J^1,J^2](u) = \ord_X^1(u)$ for any $u \in X^1$ and
$f[J^1,J^2](u) = \ord_X^2(u)$ for any $u \in X^2$;
  \label{p:layer-one-last}
\item among all functions $f$ satisfying the previous conditions,\label{p:layer-one-min}
  $f[J^1,J^2]$ minimizes the cardinality of $\sol^f$ 
(where in the expression $\sol^f$ the function $f$ is treated as an ordering of the set $(A^2 \cup X^2) \setminus A^1$
 in the \spic{} instance $(G,k,(\pos_u)_{u \in V(G)},\Gdown,\Gup)$);
\item among all functions $f$ satisfying the previous conditions, 
$f[J^1,J^2]$ is lexicographically minimum.
\end{enumerate}
\end{definition}

We first observe the following consequence of the above definition.
\begin{lemma}\label{lem:layer-one-1}
For any $p^1 \leq p^2$ such that $J^\ord(p^1),J^\ord(p^2) \in \jumpfam$, we have that $(J^\ord(p^1),J^\ord(p^2))$ is a layer-one state and
$$f[J^\ord(p^1),J^\ord(p^2)] = \ord|_{\secp{\jump{p^2}} \setminus \secp{p^1}}.$$
In particular, $\ord=f[J^\ord(1),J^\ord(n-2)]\cup \{(\omega_3,n)\}$.
\end{lemma}
\begin{proof}
Let $M:=\secp{\jump{p^2}} \setminus \secp{p^1}$. It is straightforward to verify that $(J^\ord(p^1),J^\ord(p^2))$ is a layer-one state
and $\ord|_M$ satisfies the first \ref{p:layer-one-last} properties of the value of a layer-one state.
Also, no edges of $\sol^\ord$ are incident to $X_1$ nor to $X_{n-2}$, and hence $J^\ord(1),J^\ord(n-2) \in \jumpfam$ and
$(J^\ord(1),J^\ord(n-2))$ is a layer-one state.

Let $f$ be any function satisfying the first \ref{p:layer-one-min} conditions of the definition
of a value of the layer-one state $(J^\ord(p^1),J^\ord(p^2))$.
Let $\ord'$ be an ordering of $V(G)$ defined as $\ord'(u) = f(u)$ if $u$ is the domain of $f$, and $\ord'(u) = \ord(u)$ otherwise.
It is straightforward to verify that $\ord'$ is feasible, using the separation property provided by Lemma~\ref{lem:jump-cut} and the fact that $\ord'|_{X^1\cup X^2}=\ord|_{X^1\cup X^2}$.
For the same reasons, by the definition of $\ord'$ we have that $\sol^{\ord'}=\left(\sol^{\ord}\setminus \binom{M}{2}\right)\cup \sol^f$. By the optimality of $f$ we have that $|\sol^f|\leq |\sol^{\ord|_{M}}|\leq |\sol^{\ord}\cap \binom{M}{2}|$, and so $|\sol^{\ord'}|\leq |\sol^\ord|$. By the optimality of $\ord$ we infer that $|\sol^{\ord'}|=|\sol^{\ord}|$, and $\sol^{\ord'}$ is also a minimum completion of $G$.
Since $\sol^\ord$ is also lexicographically minimum, it is easy to see that the last criterion of the definition of the value of the layer-one state $(J^\ord(p^1),J^\ord(p^2))$ indeed chooses $\ord|_{M}$.
\end{proof}
By Lemma~\ref{lem:layer-one-1}, our goal is to compute $f[J^\ord(1),J^\ord(n-2)]$ by dynamic programming. Observe
that both $J^\ord(1)$ and $J^\ord(n-2)$ are known, due to the augmentation performed at the beginning of this section.

Our dynamic programming algorithm computes value $g[J^1,J^2]$ for every
layer-one state $(J^1,J^2)$, and we will ensure that $g[J^\ord(p^1),J^\ord(p^2)] = f[J^\ord(p^1),J^\ord(p^2)]$ for any $p^1 \leq p^2$ with $J^\ord(p^1),J^\ord(p^2) \in \jumpfam$;
we will not necessarily guarantee that the values of $f$
and $g$ are equal for other states.
(Formally, $g[J^1,J^2]$ may also take value of $\bot$, which implies that either $J^1$ or $J^2$ is not consistent with $\ord$;
 we assign this value to $g[J^1,J^2]$ whenever we find no candidate for its value.)

Consider now one layer-one state $(J^1,J^2)$ with 
$J^1 = (A^1,X^1,\ord_X^1)$, $J^2 = (A^2,X^2,\ord_X^2)$.
The base case for computing $g[J^1,J^2]$ is the case where $A^2 \subseteq A^1 \cup X^1$.
Then $\ord_X^1 \cup \ord_X^2$ is the only candidate for the value $f[J^1,J^2]$, provided that $\ord_X^1$ and $\ord_X^2$ agree on the intersection of their domains.

In the other case, we iterate through all possible tuples
$J^3 = (A^3,X^3,\ord_X^3)$, with $A^1 \subset A^3 \subset A^2$
such that both $(J^1,J^3)$ and $(J^3,J^2)$ are layer-one states,
and try $g[J^1,J^3] \cup g[J^3,J^2]$ as a candidate value for $g[J^1,J^2]$.
That is, we temporarily
pick $g[J^1,J^2]$ with the same criteria as for $f[J^1,J^2]$, but taking into
account only values $g[J^1,J^3] \cup g[J^3,J^2]$ for different choices of $J^3$.

Since the minimization for $g[J^1,J^2]$ is taken over smaller set of functions
than for $f[J^1,J^2]$, we infer that
\begin{enumerate}
\item the cardinality of $\sol^{f[J^1,J^2]}$ is not larger than the cardinality of $\sol^{g[J^1,J^2]}$;
\item even if these two sets are of equal size, $f[J^1,J^2]$ is lexicographically
not larger than $g[J^1,J^2]$.
\end{enumerate}
However, observe that if $J^1 = J^\ord(p^1)$ and $J^2 = J^\ord(p^2)$
and there exists $p^3$ such that $p^1 < p^3 < p^2$ and $J^\ord(p^3) \in \jumpfam$,
    then $g[J^1,J^\ord(p^3)]\cup g[J^\ord(p^3),J^2]$ is taken into account when evaluating $g[J^1,J^2]$. If we compute the values for the states $(J^1,J^2)$ in the order of non-decreasing values of $|A^2\setminus A^1|$, then the values $g[J^1,J^\ord(p^3)], g[J^\ord(p^3),J^2]$ have been computed before, and moreover by induction hypothesis they are equal to $f[J^1,J^\ord(p^3)]$ and $f[J^\ord(p^3),J^2]$, respectively. Therefore,
$$f[J^1,J^\ord(p^3)] \cup f[J^\ord(p^3),J^2] = \ord|_{\secp{\jump{p^2}} \setminus \secp{p^1}}$$
is taken as a candidate value for $g[J^1,J^2]$ and, consequently, $g[J^1,J^2] = f[J^1,J^2] = \ord_{\secp{\jump{p^2}} \setminus \secp{p^1}}$.

Finally, we need to ensure that $g[J^1,J^2] = f[J^1,J^2]$ in the case
when such position $p^3$ does not exist. To this end, we take also more
candidate values for $g[J^1,J^2]$, computed by the layer-two dynamic programming
in the next section. We ensure that, if $J^1 = J^\ord(p^1)$,
$J^2 = J^\ord(p^2)$ but for any $p^1 < q < p^2$ we have $J^\ord(q) \notin \jumpfam$,
then the layer-two dynamic programming actually outputs $f[J^1,J^2]$ as one of the candidates,
     and runs in time $(n|\secfam|)^{\Oh(\tau)}$ for any choice of $J^1,J^2$. By Theorem~\ref{thm:jump-enum} there are at most $n^{\Oh(k/\tau)}|\secfam|^4$ layer-one states. Hence by using $(n|\secfam|)^{\Oh(\tau)}$ work for each of them will give the running time promised in Theorem~\ref{thm:dp}.

\subsection{Layer two: chains}

In this section we are given a layer-one state $(J^1,J^2)$ with $J^1 = (A^1,X^1,\ord_X^1)$, $J^2 = (A^2,X^2,\ord_X^2)$;
denote $p^\alpha = |A^\alpha| + 1$, $r^\alpha = |A^\alpha \cup X^\alpha|+1$ for $\alpha = 1,2$.
We are to compute, in time $(n|\secfam|)^{\Oh(\tau)}$, the value $f[J^1,J^2]$, assuming: $J^1 = J^\ord(p^1)$, $J^2 = J^\ord(p^2)$, and
for no position $p^1 < q < p^2$ it holds that $J^\ord(q) \in \jumpfam$. By Theorem~\ref{thm:jump-enum}, it implies that the number of edges of $\sol^\ord$
incident to any set $\jumpset{q}$ for $p^1 < q < p^2$ is more than $2k/\tau$. Observe that the following holds  $X^\alpha=\secp{\jump{p^\alpha}} \setminus \secp{p^\alpha}$, and hence $r^\alpha=\jump{p^\alpha}$ for $\alpha=1,2$. 

For any position $q$, consider the following sequence: $z_q(0) = q$ and $z_q(i+1) = \jump{z_q(i)}$ (with the convention that $\jump{\infty} = \infty$).
Observe the following.
\begin{lemma}\label{lem:short-z}
For any $q \geq p^1$ it holds that $z_q(\tau) \geq p^2$.
\end{lemma}
\begin{proof}
Consider any $q \geq p^1$. For any $i > 0$ such that $z_q(i) < p^2$ we have
that there are more than $2k/\tau$ edges of $\sol^\ord$ incident to $\jumpset{z_q(i)}$.
However, the sets $\jumpset{z_q(i)}$ are pairwise disjoint for different values of $i$.
Since $|\sol^\ord|\leq k$, we infer that for less than $\tau$ values $i > 0$ we may have $z_q(i) < p^2$, and the lemma is proven.
\end{proof}

By a straightforward induction from Lemma~\ref{lem:jump-ineq} we obtain the following.
\begin{lemma}\label{lem:z-interlace}
For any two positions $c,d$ with $c \leq d \leq \jump{c}$
and for any $i \geq 0$ it holds that
$$z_c(i) \leq z_d(i) \leq z_c(i+1).$$
\end{lemma}
The next observation gives us the crucial separation property for the layer-two dynamic programming (see also Figure~\ref{fig:CiDi}).
\begin{lemma}\label{lem:z-cut}
For any positions $c,d$ with $c \leq d \leq \jump{c}$
define
\begin{align*}
C_i &= \ord^{-1}([z_c(i),z_d(i)-1]),\\
D_i &= \ord^{-1}([z_d(i),z_c(i+1)-1]).
\end{align*}
Then
\begin{enumerate}
\item sets $C_i,D_i$ form a partition of $V(G)\setminus A_c$;
\item for any $i \geq 0$, it holds that both $C_i \cup D_i$ and $D_i \cup C_{i+1}$
are cliques in $G^\ord$;
\item for any $j > i \geq 0$ there is no edge in $G^\ord$ between $C_i$ and $D_j$;
\item for any $i > j+1 > 0$ there is no edge in $G^\ord$ between $C_i$ and $D_j$.
\end{enumerate}
\end{lemma}
\begin{proof}
All statements follow from the definitions $z_c(i+1) = \jump{z_c(i)}$
and $z_d(i+1) = \jump{z_d(i)}$, and from Lemmata~\ref{lem:jump-cut} and~\ref{lem:z-interlace}.
\end{proof}

\begin{figure}
\centering
\includegraphics{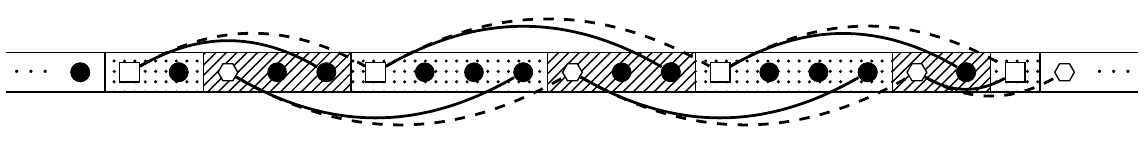}
\caption{The separation property provided by Lemma~\ref{lem:z-cut}.
The sequences $z_c(i)$ and $z_d(i)$ are denoted with rectangular and hexagonal shapes, respectively.
The sets $C_i$ and $D_i$ are denoted with dots and lines, respectively.}
\label{fig:CiDi}
\end{figure}

Intuitively, Lemma~\ref{lem:z-cut} implies that we may independently consider the vertices of
$\bigcup_{i \geq 0} C_i$ and of $\bigcup_{i \geq 0} D_i$: the sequences $z_c(i)$ and $z_d(i)$
give us some sort of `horizontal' partition of the graphs $G$ and $G^\ord$.
We now formalize this idea.

\begin{definition}[chain]
A \emph{chain} is a quadruple $(s,z,u,B)$, where
\begin{align*}
s &\in \{0,1,\ldots,\tau\}, \\
z &: \{0,1,\ldots,s\} \to [p^1,r^2], \\
u &: \{0,1,\ldots,s\} \to V(G), \\
B &: \{0,1,\ldots,s\} \to 2^{V(G)}
\end{align*}
with the following properties:
\begin{enumerate}
\item 
  $z(i) \in [p^2,r^2]$ if and only if $i = s$;
\item $z(i) < z(i+1)$ for any $0 \leq i < s$;
\item $|B(i)| = z(i)-1$ for any $0 \leq i \leq s$;
\item $B(i) \subset B(i+1)$, for any $0 \leq i < s$;
\item $u(i) \in B(j)$ if and only if $0 \leq i < j \leq s$;
\item no edge of $G$ connects a vertex of $B(i)$ with a vertex
of $V(G) \setminus B(i+1)$ for any $0 \leq i < s$.
\end{enumerate}
A chain $(s,z,u,B)$ is \emph{consistent} with the ordering $\ord$
if $s = \min\{i: z_{z(0)}(i) \geq p^2\}$ and for all $0 \leq i \leq s$
\begin{enumerate}
\item $z(i) = z_{z(0)}(i)$;
\item $\ord(u(i)) = z(i)$;
\item $B(i) = \secp{z(i)}$.
\end{enumerate}
\end{definition}

We remark here that if $p^2 \leq n-2$ then $\jump{q} \leq r^2$ for any $q \leq p^2$, and hence $z_{z(0)}(s) \leq r^2$ 
for any $z(0) \leq r^2$ in the definition above.

Our next lemma follows immediately  from the definition of a chain.
\begin{lemma}\label{lem:ord-to-chain}
For  $q \in [p^1,r^2]$, let $s = \min \{i: z_q(i) \geq p^2\}$. For every $0 \leq i \leq s$, let 
\begin{align*}
z(i)  =  z_q(i),\\
u(i) = \ord^{-1}(z(i)) ,\\
B(i) = \secp{z(i)} .
\end{align*}
Then $I^\ord(q) := (s,z,u,B)$ is a chain consistent with $\ord$.

\end{lemma}

Moreover, the bound of Lemma~\ref{lem:short-z} gives us the following enumeration algorithm.
\begin{theorem}\label{thm:chain-enum}
In $(n|\secfam|)^{\Oh(\tau)}$ time one can enumerate a family $\chainfam$
of at most $(n|\secfam|)^{\Oh(\tau)}$ chains that contains all chains consistent with $\ord$.
\end{theorem}
\begin{proof}
There are $1+\tau \leq n$ possible values for $s$.
For each $0 \leq i \leq s$, there are at most $n$ choices for $z(i)$,
$n$ choices for $u(i)$ and $|\secfam|$ choices for $B(i)$.
The bound $s \leq \tau$ due to Lemma~\ref{lem:short-z} yields the desired bound.
Observe that the properties of a chain can be verified in polynomial time.
\end{proof}

We are now finally ready to state the definition of a layer-two state with its value.
\begin{definition}[layer-two state]
A \emph{layer-two state} consists of two chains $I^1 = (s^1,z^1,u^1,B^1)$, $I^2 = (s^2,z^2,u^2,B^2)$ with $I^1,I^2 \in \chainfam$ such that
\begin{enumerate}
\item $s^2 \leq s^1 \leq s^2 + 1$,
\item $z^1(i) \leq z^2(i)$, $B^1(i) \subseteq B^2(i)$ for any $1 \leq i \leq s^2$ and $z^2(i) \leq z^1(i+1)$, $B^2(i) \subseteq B^1(i+1)$ for any $1 \leq i < s^1$;
\item $u^1(i) = u^2(j)$ if and only if $z^1(i) = z^2(j)$ for any $1 \leq i \leq s^1$ and $1 \leq j \leq s^2$;
\end{enumerate}
Furthermore, we denote
\begin{align*}
C_i[I^1,I^2] &= B^2(i) \setminus B^1(i) &\mathrm{for\ any\ }0 \leq i\leq s^2,\\
D_i[I^1,I^2] &= B^1(i+1) \setminus B^2(i)&\mathrm{for\ any\ }0 \leq i < s^1,\\
Z_i[I^1,I^2] &= [z^1(i),z^2(i)-1] & \mathrm{for\ any\ }0 \leq  i \leq s^2,\\
C_{s^1}[I^1,I^2] &= (A^2 \cup X^2) \setminus B^1(s^1) &\mathrm{if\ }s^2 < s^1,\\
Z_{s^1}[I^1,I^2] &= [z^1(s^1), r^2-1] &\mathrm{if\ }s^2 < s^1,\\
\end{align*}
and require that for any $0 \leq i \leq s^1$ all positions of $Z_i[I^1,I^2]$ are pairwise adjacent in $\Gup$.
We define $\Gdown^\ast$ to be equal to $\Gdown$ with additionally $[p^2,r^2-1]$ and each $Z_i[I^1,I^2]$ turned into a clique, for every $0 \leq i \leq s^1$.
Note that by Lemma~\ref{lem:intersection-union}, $\Gdown^\ast$ is a proper interval graph with identity being an umbrella ordering. Moreover, it holds that $E(\Gdown^\ast) \subseteq E(\Gup)$ by the construction of $E(\Gdown^\ast)$ and the fact that $J^2\in \jumpfam$.

The \emph{value} of a layer-two state $(I^1,I^2)$ is a bijection
$f[I^1,I^2]: \bigcup_{i=0}^{s^1} C_i[I^1,I^2] \to \bigcup_{i=0}^{s^1} Z_i[I^1,I^2]$
such that:
\begin{enumerate}
\item $f[I^1,I^2]$ is a feasible ordering of its domain, that is,
for any $u$ in the domain of $f[I^1,I^2]$ we have $f[I^1,I^2](u) \in \pos_u$, and
  for any $u_1,u_2$ in the domain of $f[I^1,I^2]$ such that $u_1u_2\in E(G)$, it holds that
  $f[I^1,I^2](u_1)f[I^1,I^2](u_2) \in E(\Gup)$;
\item $f[I^1,I^2](u) \in Z_i[I^1,I^2]$ whenever $u \in C_i[I^1,I^2]$;
\item $f[I^1,I^2](u^1(i)) = z^1(i)$ for all $0 \leq i \leq s^1$;
\item $f[I^1,I^2](u) = \ord_X^1(u)$ whenever $u \in X^1$ and $f[I^1,I^2](u) = \ord_X^2(u)$ whenever $u \in X^2$;\label{p:layer-two-last}
\item among all functions $f$ satisfying the previous conditions, $f[I^1,I^2]$ minimizes the cardinality
of $\sol^{f,\ast}$,
where the set $\sol^{f,\ast}$ is defined as the unique minimal completion for the ordering $f$ of the subgraph of $G$ induced by the domain of $f$
and \spic{} instance $(G,k,(\pos_u)_{u \in V(G)},\Gdown^\ast,\Gup)$;\label{p:layer-two-min}
\item among all functions $f$ satisfying the previous conditions, 
$f[I^1,I^2]$ is lexicographically minimum.
\end{enumerate}
\end{definition}

Note that in the definition of a layer-two state we {\em{do not}} require that any of the chains begins in $[p^1,r^1]$, i.e. that $z^1(0)$ or $z^2(0)$ are in this interval. The values for the states where these chains begin at arbitrary positions within $[p^1,r^2]$ will be essential for computing the final value we are interested in.

Similarly as in the case of layer-one states, we have the following claim.
\begin{definition}[relevant pair]
A pair $(q^1,q^2)$ with $p^1 \leq q^1 \leq q^2 \leq \min(\jump{q^1},r^2)$ is called
\emph{relevant} if one of the following holds:
\begin{enumerate}
\item $q^2 \leq r^1$,
\item $q^1 = q^2$, or
\item there exists a position $q_{\leftarrow} \geq p^1$
such that $\jump{q_\leftarrow} \leq q^1 \leq q^2 \leq \jump{q_\leftarrow+1}$
(see also Figure~\ref{fig:relevant}).
\end{enumerate}
\end{definition}

\begin{figure}
\centering
\includegraphics{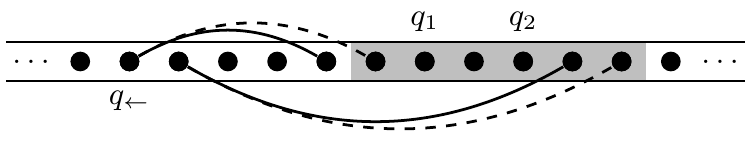}
\caption{The last case of the definition of a relevant pair $(q_1,q_2)$.
  Both positions $q_1$ and $q_2$ need to belong to the gray area.}
\label{fig:relevant}
\end{figure}

\begin{lemma}\label{lem:layer-two-ord}
For any $q_1,q_2$ such that $p^1 \leq q^1 \leq q^2 \leq \min(\jump{q^1},r^2)$, the pair $(I^\ord(q^1),I^\ord(q^2))$ is a layer-two state.
If moreover $(q^1,q^2)$ is a relevant pair, then $f[I^\ord(q^1),I^\ord(q^2)]$ is a restriction of $\ord$ to the domain of $f[I^\ord(q^1),I^\ord(q^2)]$.
In particular, $f[I^\ord(p^1),I^\ord(r^1)] = f[J^\ord(p^1),J^\ord(p^2)]$.
\end{lemma}
\begin{proof}
By somehow abusing the notation, we denote $X^1=X_{p^1}$ and $X^2=X_{p^2}$. It is straightforward to verify from the definition that $(I^\ord(q^1),I^\ord(q^2))$ is a layer-two state
and the restriction of $\ord$ to $Y := \bigcup_i C_i[I^\ord(q^1),I^\ord(q^2)]$ satisfies the first~\ref{p:layer-two-last} requirements
of the definition of a value of a layer-two state, even if $(q^1,q^2)$ is not a relevant pair.
Moreover, observe that Lemma~\ref{lem:z-cut} implies that $\sol^\ord \cap \binom{Y}{2}$ is a completion for the ordering
$\ord|_Y$ of $Y$ in the instance $(G,k,(\pos_u)_{u \in V(G)},\Gdown^\ast,\Gup)$. Hence, $\sol^{\ord|_Y,\ast} \subseteq \sol^\ord \cap \binom{Y}{2}$.

Now assume that $(q^1,q^2)$ is a relevant pair and denote $f = f[I^\ord(q^1),I^\ord(q^2)]$ and $I^\ord(q^\alpha) = (s^\alpha,z^\alpha,u^\alpha,B^\alpha)$
for $\alpha=1,2$.
If $q^1 = q^2$, then observe that the sets $C_i[I^\ord(q^1),I^\ord(q^2)]$ are empty, and the state in question asks for an empty function. Hence, assume $q^1 < q^2$.
Define an ordering $\ord'$ of $V(G)$ so that $\ord'(u) = f(u)$ for any $u\in Y$, and $\ord'(u) = \ord(u)$ otherwise.

Let us define $\sol := (\sol^\ord \setminus \binom{Y}{2}) \cup \sol^{f,\ast}$. 
In the subsequent claims we establish some properties of the graph $G+\sol$ and ordering $\ord'$.
\begin{claim}\label{cl:layer-two:margins}
$\ord'(G+\sol)[[1,r^1-1]\cup[p^2,n]] = \ord(G^\ord)[[1,r^1-1]\cup[p^2,n]]$.
\end{claim}
\begin{proof}
Note here that $\ord'$ and $\ord$ agree on positions before $r^1$ and after $p^2-1$. Observe also that $[p^1,r^1-1]$ and $[p^2,r^2-1]$ are cliques in $\ord(G^\ord)$, and $[p^1,r^1-1]$ can have non-empty intersection only with the first of the intervals $Z_i[I^\ord(q^1),I^\ord(q^2)]$.
Since $\binom{[p^2,r^2-1]}{2},\binom{Z_i[I^\ord(q^1),I^\ord(q^2)]}{2}\subseteq E(\Gdown^\ast)\cap \binom{\sigma(Y)}{2}\subseteq E(\sigma'(G[Y]+\sol^{f,\ast}))$, it follows by the definition of $\sol$ that that intervals $[p^1,r^1-1]$ and $[p^2,r^2-1]$ are cliques in $\ord'(G+\sol)$ as well. Since $Y\subseteq \ord^{-1}([p^1,r^2-1])$, the claim follows.
\cqed\end{proof}

\begin{claim}\label{cl:layer-two:umbrella}
$\ord'$ is a umbrella ordering of $G+\sol$.
\end{claim}
\begin{proof}
Consider any $a,b,c \in V(G)$ with $ac \in E(G+\sol)$ and
$\ord'(a) < \ord'(b) < \ord'(c)$; we want to show the umbrella property for the triple
$a,b,c$ in the graph $G+\sol$. We consider a few cases, depending on the intersection $\{a,b,c\} \cap Y$.
\begin{enumerate}
\item If $a,b,c \in Y$ or $a,b,c \notin Y$, then the umbrella property holds by the definition of $\sol^\ord$ and $\sol^{f,\ast}$.
\item If $\ord'(a) \geq p^2$ or $\ord'(c) < r^1$, then recall that $\ord'(G+\sol)[[1,r^1-1]\cup[p^2,n]] = \ord(G^\ord)[[1,r^1-1]\cup[p^2,n]]$. Then the umbrella property for $a,b,c$ follows from the fact that $\ord$ is an umbrella ordering of $G^\ord$.

Hence, in the remaining cases we have in particular that $a\notin X^2$ and $c\notin X^1$. Observe also that the assumption $ac \in E(G+\sol)$ implies that $p^1 \leq \ord'(a) < \ord'(c) < r^2$, since $r^1=\jump{p^1}$ and $r^2=\jump{p^2}$.

\item If $a,c \in Y$ and $b \notin Y$ then, by the structure of $Y$, we have $a \in C_i[I^\ord(q^1),I^\ord(q^2)]$, $c \in C_j[I^\ord(q^1),I^\ord(q^2)]$ for some $1 \leq i < j \leq s^1$. We claim that $j = i+1$. Assume the contrary. Observe that if $i+1 < j$ then in particular $i < s^2$.
By Lemma~\ref{lem:z-cut}, no edge of $G^\ord$ connects $\secp{z_{q^2}(i)}$ with $V(G) \setminus\secp{z_{q^2}(i+1)}$, so in particular there is no such edge neither in $G$, which is subgraph of $G^\ord$. Likewise, there is no edge between $[1,z_{q^2}(i)]$ and $[z_{q^2}(i+1),n]$ in $\Gdown^\ast$. By the construction of $\sol^{f,\ast}$ it follows that also no edge of $\sol^{f,\ast}$ connects $\secp{z_{q^2}(i)}$ with $V(G) \setminus\secp{z_{q^2}(i+1)}$.
As $\ord$ and $\ord'$ differ only on the internal ordering of each set $C_i[I^\ord(q^1),I^\ord(q^2)]$, and $ac\in E(G+F)$, we have a contradiction,
and hence $c \in C_{i+1}[I^\ord(q^1),I^\ord(q^2)]$.
It follows that $b \in D_i[I^\ord(q^1),I^\ord(q^2)]$ and, by Lemma~\ref{lem:z-cut}, $ab,bc \in E(G^\ord)$.
  By the definition of $\sol$, $ab,bc \in E(G)\cup \sol$.

In the remaining cases we have that either $a$ or $c$ does not belong to $Y$. Hence $ac\in E(G^\ord)$ by the definition of $F$.

\item If $a \in Y \setminus X^2$ and $c \notin Y$, then, by Lemma~\ref{lem:z-cut}, we have $a \in C_i[I^\ord(q^1),I^\ord(q^2)]$
and $c \in D_i[I^\ord(q^1),I^\ord(q^2)]$ for some $0  \leq i < s^1$. By Lemma~\ref{lem:z-cut},
$C_i[I^\ord(q^1),I^\ord(q^2)] \cup D_i[I^\ord(q^1),I^\ord(q^2)]$ is a clique in $G^\ord$, and, by the definition of $\Gdown^\ast$,
$C_i[I^\ord(q^1),I^\ord(q^2)]$ is a clique in $G+\sol$. Hence, $ab,bc \in E(G) \cup \sol$ regardless whether $b\in Y$ or not.
\item If $a \notin Y$ and $c \in C_i[I^\ord(q^1),I^\ord(q^2)]$ for some $i > 0$, then, by Lemma~\ref{lem:z-cut},
  $a \in D_{i-1}[I^\ord(q^1),I^\ord(q^2)]$. As in the previous case, Lemma~\ref{lem:z-cut} asserts that 
$D_{i-1}[I^\ord(q^1),I^\ord(q^2)] \cup C_i[I^\ord(q^1),I^\ord(q^2)]$ is a clique in $G^\ord$, and the definition of $\Gdown^\ast$
gives us that $C_i[I^\ord(q^1),I^\ord(q^2)]$ is a clique in $G+\sol$. Consequently, $ab,bc \in E(G) \cup \sol$  regardless whether $b\in Y$ or not.
\item If $a \notin Y$ and $c \in C_0[I^\ord(q^1),I^\ord(q^2)] = \ord^{-1}([q^1,q^2-1])$ then, as $\ord'(c) \geq r^1$, we have that pair $(q^1,q^2)$ is 
a relevant pair due to existence of some position $q_\leftarrow$. Since $ac\in E(G^\ord)$, we have that $\ord'(a) = \ord(a) \geq q_\leftarrow+1$.
As $\jump{q_\leftarrow+1} \geq q^2$, we have that also $ab\in E(G^\ord)$ and $bc\in E(G^\ord)$. By the definition of $F$ we infer that $ab \in E(G)\cup \sol$ and, additionally, $bc \in E(G) \cup \sol$ in the case $b \notin Y$.
If $b \in Y$ then $b \in C_0[I^\ord(q^1),I^\ord(q^2)]$ and $bc \in E(G) \cup \sol$ by the definition of $\Gdown^\ast$.
\item If $a,c \notin Y$ and $b \in Y$, then let $b \in C_i[I^\ord(q^1),I^\ord(q^2)]$ for some
$0 \leq i \leq s^1$. Since $\ord$ and $\ord'$ differ only on internal ordering of sets $C_i[I^\ord(q^1),I^\ord(q^2)]$ and $a,c\notin Y$, then the condition $\ord'(a) < \ord'(b) < \ord'(c)$ implies also $\ord(a) < \ord(b) < \ord(c)$. Since $ac\in E(G^\ord)$ and $\ord$ is an umbrella ordering of $G^\ord$, we infer that $ab,bc\in E(G^\ord)$. By the definition of $F$ this implies that $ab,bc\in E(G+F)$.
\end{enumerate}
\cqed\end{proof}

\begin{claim}\label{cl:layer-two:Gdown}
$E(\Gdown^\ast) \subseteq E(\ord'(G+\sol))$.
\end{claim}
\begin{proof}
Consider any $pq \in E(\Gdown^\ast)$.
Denote $a = \ord^{-1}(p)$, $b = \ord^{-1}(q)$ and similarly denote $a'$ and $b'$ for the ordering $\ord'$; we want to show that $a'b' \in E(G) \cup \sol$.
As $E(\Gdown^\ast) \subseteq E(\ord(G^\ord))$ we have $ab \in E(G^\ord)$.
If $p,q \in \bigcup_i Z_i[I^\ord(q^1),I^\ord(q^2)]$ then $a'b' \in E(G) \cup \sol^{f,\ast}$ by the definition of $\sol^{f,\ast}$.
Otherwise, without loss of generality assume that $q \notin \bigcup_i Z_i[I^\ord(q^1),I^\ord(q^2)]$, and hence $b = b'$.
If additionally $a=a'$ then $a'b' \in E(G) \cup\sol$ follows directly from the definition of $\sol$ and the fact that $ab \in E(G^\ord)$.
In the remaining case, if $a \neq a'$, we have $p \in Z_i[I^\ord(q^1),I^\ord(q^2)]$ and $a,a' \in C_i[I^\ord(q^1),I^\ord(q^2)]$ for some
$0 \leq i \leq s^1$. Moreover, from the assumption $a \neq a'$ we infer that $r^1 \leq p < p^2$, and consequently $i < s^1$.
By the definition of $\sol$, we need to show that $a'b \in E(G^\ord)$.

We consider two cases, depending on the relative order of $p$ and $q$. If $p < q$, then we have $z^2(i) \leq q < z^1(i+1)$ by Lemma~\ref{lem:z-cut}
and consequently $b \in D_i[I^\ord(q^1),I^\ord(q^2)]$. By Lemma~\ref{lem:z-cut} again, $b$ is adjacent to all vertices of $C_i[I^\ord(q^1),I^\ord(q^2)]$
in the graph $G^\ord$, and $a'b \in E(G^\ord)$.
A similar argument holds if $q < p$ and $i > 0$: by Lemma~\ref{lem:z-cut}, we have first that $b \in D_{i-1}[I^\ord(q^1),I^\ord(q^2)]$ and, second,
that $b$ is adjacent in $G^\ord$ to all vertices of $C_i[I^\ord(q^1),I^\ord(q^2)]$, and hence $a'b \in E(G^\ord)$.
In the remaining case, if $q <p$ and $i=0$ (hence $p\in [q^1,q^2-1]$),  from $p \geq r^1$ it follows that the reason why $(q^1,q^2)$ is a relevant pair
is existence of some position $q_\leftarrow$. Since $ab \in E(G^\ord)$, we infer that $q \geq q_\leftarrow+1$. Hence, $b$ is
adjacent in $G^\ord$ to all vertices of $C_0[I^\ord(q^1),I^\ord(q^2)]$, in particular to $a'$, and the claim is proven.  
\cqed\end{proof}

\begin{claim}\label{cl:layer-two:Gup} 
$E(\ord'(G)) \subseteq \Gup$ and $\ord'$ is a feasible ordering of $G$.
\end{claim}
\begin{proof}
Observe that it follows directly from the definition of $\ord'$ that $\ord'(u) \in \pos_u$ for any vertex $u$.
Hence, to show feasibility of $\ord'$ it suffices to show that $E(\ord'(G)) \subseteq \Gup$.

Consider any $ab \in E(G)$. If both $a$ and $b$ belong to $Y$ or both do not belong,
then the claim is obvious by the feasibility of both $\ord$ and $f$.
Assume then $a \in Y$ and $b \notin Y$. If $\ord(a) = \ord'(a)$ then clearly $\ord'(a)\ord'(b) = \ord(a)\ord(b) \in E(\Gup)$. 
Otherwise, $a \notin X^2$ and $a \in C_i[I^\ord(q^1),I^\ord(q^2)]$ for some $0 \leq i < s^1$.
If $\ord(b) \geq z_{q^2}(i)$ then Lemma~\ref{lem:z-cut} implies that $b \in D_i[I^\ord(q^1),I^\ord(q^2)]$.
By Lemma~\ref{lem:z-cut} again, $[z_{q^1}(i),z_{q^1}(i+1)-1]$ is a clique in $\ord(G^\ord)$ and hence in $\Gup$ as well, so $\ord'(a)\ord'(b) \in E(\Gup)$.
A similar situation happens if $\ord(b) < z_{q^1}(i)$ and $i>0$: $b \in D_{i-1}[I^\ord(q^1),I^\ord(q^2)]$ and again Lemma~\ref{lem:z-cut}
together with feasibility of $\ord$ proves the claim.
In the remaining case $i = 0$ and $\ord(b) < q^1$. As $\ord(a) \neq \ord'(a)$ we have $a \notin X^1$ and hence the reason why $(q^1,q^2)$ is a relevant pair
must be existence of some position $q_\leftarrow$. As $ab \in E(G)$ we have $\ord(b) \geq q_\leftarrow+1$. As $\jump{q_\leftarrow+1} \geq q^2$,
the position $b$ is adjacent to all positions of $[q^1,q^2-1]$ in $\ord(G^\ord)$ and hence $\ord'(a)\ord'(b) \in E(\ord(G^\ord)) \subseteq E(\Gup)$ as claimed.
\cqed\end{proof}

From the above claims we infer that $|\sol^{\ord'}| \leq |\sol^{f,\ast}| + |\sol^\ord \setminus \binom{Y}{2}|$, whereas
$\sol^{\ord|_Y,\ast} \subseteq \sol^\ord \cap \binom{Y}{2}$.
By the minimality of both $f$ and $\ord$, including the lexicographical minimality, we have $f = \ord|_Y$ and the lemma is proven.
\end{proof}

The layer-two dynamic programming algorithm computes,
for any layer-two state $(I^1,I^2)$,   
a function $g[I^1,I^2]$
that satisfies the first~\ref{p:layer-two-last} conditions of $f[I^1,I^2]$, and we will inductively ensure that $g[I^\ord(q^1),I^\ord(q^2)] = f[I^\ord(q^1),I^\ord(q^2)]$
for any relevant pair $(q^1,q^2)$.
We compute the values $g[I^1,I^2]$ in the order of decreasing value of $z^1(0)$ and, subject to that, increasing value of $z^2(0)$.
(Formally, $g[I^1,I^2]$ may also take value of $\bot$, which implies that either $I^1$ or $I^2$ is not consistent with $\ord$;
 we assign this value to $g[I^1,I^2]$ whenever we find no candidate for its value.)

Consider now a fixed layer-two state $(I^1,I^2)$ with
$I^1 = (s^1,z^1,u^1,B^1)$ and $I^2 = (s^2,z^2,u^2,B^2)$.
We start with the the base case when we have that either $z^1(0) = z^2(0)$ or
$z^1(0) \geq p^2-1$. Observe that in this situation we have that the domain of $g[I^1,I^2]$ is either $X^2$ or $X^2$ with an additional element $u^1(0)$ which must be mapped to $z^1(0)=p_2-1$. Hence all the values of $f[I^1,I^2]$ are fixed by $\ord_X^2$, $u^1$ and $z^1$, and there is only one candidate for this value.
It is straightforward to verify that, in the case when $I^1 = I^\ord(q^1)$ and $I^2 = I^\ord(q^2)$, this
unique candidate is indeed a restriction of $\ord$ and hence equals $f[I^1,I^2]$.

In the inductive step we have $z^1(0) < z^2(0)$ and $z^1(0) < p^2-1$.
We consider two cases, depending on the value of $z^2(0) - z^1(0)$.

First assume $z^2(0) - z^1(0) > 1$.
In this case consider all possible chains $I^3 = (s^3,z^3,u^3,B^3)$ such that both $(I^1,I^3)$ and $(I^3,I^2)$ are layer-two states,
and $z^1(0) < z^3(0) < z^2(0)$.
We take as candidate value for $g[I^1,I^2]$ the union $g[I^1,I^3] \cup g[I^3,I^2]$, and pick $g[I^1,I^2]$ using the criteria
from the definition of the value $f[I^1,I^2]$, but taking only functions $g[I^1,I^3] \cup g[I^3,I^2]$ for all choices of $I^3$ as candidates.

We claim that if $I^1 = I^\ord(q^1)$, $I^2 = I^\ord(q^2)$ and $(q^1,q^2)$ is a relevant pair, then 
$g[I^1,I^2] = f[I^1,I^2]$. Note that it suffices to show that $f[I^1,I^2]$ is considered as a candidate for $g[I^1,I^2]$
in the aforementioned process for some choice of $I^3$.
Consider any $q^1 < q^3 < q^2$ and observe that if $(q^1,q^2)$ is a relevant pair,
then also $(q^1,q^3)$ and $(q^3,q^2)$ are relevant pairs: this is clearly true for
the case $q^2 \leq r^1$ and, in the last case of the definition of a relevant pair,
notice that the same position $q_\leftarrow$ witnesses also that $(q^1,q^3)$ and $(q^3,q^2)$ are relevant.
Denote $I^3 = I^\ord(q^3)$ and observe that we consider a candidate $g[I^1,I^3] \cup g[I^3,I^2]$ for $g[I^1,I^2]$.
By Lemma~\ref{lem:layer-two-ord} and the inductive assumption, this candidate is a restriction of $\ord$, and hence,
again by Lemma~\ref{lem:layer-two-ord}, equals $f[I^1,I^2]$.

We are left with the case $z^2(0) = z^1(0) + 1$. As $z^1(0) < p^2 - 1$, we have $z^2(0) < p^2$.
For $\alpha=1,2$ define $s^\alpha_\ast = s^\alpha-1$, and
$z^\alpha_\ast(i) = z^\alpha(i+1)$, $u^\alpha_\ast(i) = u^\alpha(i+1)$ and $B^\alpha_\ast(i) = B^\alpha(i+1)$
for any $0 \leq i \leq s^\alpha_\ast$, and $I^\alpha_\ast = (s^\alpha_\ast, z^\alpha_\ast, u^\alpha_\ast, B^\alpha_\ast)$.
In this case we consider only one candidate for $g[I^1,I^2]$, being $g[I^1_\ast,I^2_\ast]$, extended with 
$g[I^1,I^2](u^1(0)) = z^1(0)$.

It remains to show that if $I^1 = I^\ord(q^1)$, $I^2 = I^\ord(q^2)$ and $(q^1,q^2)$ is an relevant pair,
then $g[I^1,I^2] = f[I^1,I^2]$.
Observe that $I^1_\ast = I^\ord(\jump{q^1})$ and $I^2_\ast = I^\ord(\jump{q^2})$.
Moreover, the position $q^1$ witnesses that $(\jump{q^1},\jump{q^2})$ is a relevant pair, and hence
$g[I^1_\ast,I^2_\ast] = f[I^1_\ast,I^2_\ast]$ by induction.
This completes the proof that $g[I^\ord(q^1),I^\ord(q^2)] = f[I^\ord(q^1),I^\ord(q^2)]$ for all relevant pairs $(q^1,q^2)$.

As candidates for the value $f[J^1,J^2]$ of the layer-one state $(J^1,J^2)$ we are currently processing, we take all the values $g[I^1,I^2]$ for all the layer-two states $(I^1,I^2)$ for which the domain of $f[I^1,I^2]$ is equal to the domain of $f[J^1,J^2]$. By Theorem~\ref{thm:chain-enum}, there are at most $(n|\secfam|)^{\Oh(\tau)}$ guesses for such states, and they can be enumerated in $(n|\secfam|)^{\Oh(\tau)}$ time. Observe also that if indeed $J^1 = J^\ord(p^1)$ and $J^2 = J^\ord(p^2)$, then the layer-two state $(I^1,I^2) = (I^\ord(p^1),I^\ord(r^1))$ will be among the enumerated states. Since $(p^1,r^1)$ is a relevant pair, we have that $g[I^\ord(p^1),I^\ord(r^1)]=f[I^\ord(p^1),I^\ord(r^1)]$, while by Lemma~\ref{lem:layer-two-ord} we have that $f[I^\ord(p^1),I^\ord(r^1)]$ is equal to the restriction of $\ord$ to its domain, which in turn is equal to the domain of $g[J^1,J^2]$. Hence, the restriction of $\ord$ to the domain of $g[J^1,J^2]$, which is exactly equal to $f[J^1,J^2]$ by Lemma~\ref{lem:layer-one-1}, will be among the enumerated candidate values --- this was exactly the property needed by the layer-one dynamic program.

By Theorem~\ref{thm:chain-enum} there are $(n|\secfam|)^{\Oh(\tau)}$ layer-two states, thus the entire computation of $f[J^1,J^2]$
takes $(n|\secfam|)^{\Oh(\tau)}$ time, as was promised.
This concludes the proof of Theorem~\ref{thm:dp}, and hence finishes the proof of Theorem~\ref{thm:main}.

\section{Conclusions}\label{sec:conc}
We have presented the first subexponential algorithm for
\picname{}, running in time $k^{\Oh(k^{2/3})} + \Oh(nm(kn+m))$.
As many algorithms for completion problems in similar graph classes~\cite{my-ic,DrangeFPV13,FominV13,ghosh2012faster} run in time
$\Ohstar(k^{\Oh(\sqrt{k})})$, it is tempting to ask for such a running time also in our case. The bottleneck in the presented approach is the trade-offs between the two layers of our dynamic programming.

Also, observe that every $\Ohstar(2^{o(\sqrt{k})})$-time algorithm for {\sc{PIC}} would be in fact
also a $2^{o(n)}$-time algorithm. Since existence of such an algorithm seems unlikely,  we would like to ask for a $2^{\Omega(\sqrt{k})}$ lower bound, under the assumption of the
Exponential Time Hypothesis. Note that no such lower bound is known for any other completion problem for related graph classes.

\bibliographystyle{abbrv}
\bibliography{../completion}

\begin{thebibliography}{10}

\bibitem{alon2009fast}
N.~Alon, D.~Lokshtanov, and S.~Saurabh.
\newblock Fast {FAST}.
\newblock In {\em Proceedings of the 36th Colloquium of Automata, Languages and
  Programming (ICALP)}, volume 5555 of {\em Lecture Notes in Computer Science},
  pages 49--58. Springer, 2009.

\bibitem{pic-kernel}
S.~Bessy and A.~Perez.
\newblock Polynomial kernels for {P}roper {I}nterval {C}ompletion and related
  problems.
\newblock {\em Information and Computation}, 231(0):89 -- 108, 2013.

\bibitem{my-ic}
I.~Bliznets, F.~V. Fomin, M.~Pilipczuk, and M.~Pilipczuk.
\newblock A subexponential parameterized algorithm for {I}nterval {C}ompletion,
  2014.
\newblock Manuscript, submitted to arxiv.

\bibitem{Cai96}
L.~Cai.
\newblock Fixed-parameter tractability of graph modification problems for
  hereditary properties.
\newblock {\em Inf. Process. Lett.}, 58(4):171--176, 1996.

\bibitem{demaine2005subexponential}
E.~D. Demaine, F.~V. Fomin, M.~Hajiaghayi, and D.~M. Thilikos.
\newblock Subexponential parameterized algorithms on graphs of bounded genus
  and {$H$}-minor-free graphs.
\newblock {\em J.~ACM}, 52(6):866--893, 2005.

\bibitem{DrangeFPV13}
P.~G. Drange, F.~V. Fomin, M.~Pilipczuk, and Y.~Villanger.
\newblock Exploring subexponential parameterized complexity of completion
  problems.
\newblock {\em CoRR}, abs/1309.4022, 2013.
\newblock To appear in the Proceedings of STACS 2014.

\bibitem{Feige00}
U.~Feige.
\newblock Coping with the {NP}-hardness of the graph bandwidth problem.
\newblock In {\em SWAT 2000}, pages 10--19, 2000.

\bibitem{FominV13}
F.~V. Fomin and Y.~Villanger.
\newblock Subexponential parameterized algorithm for minimum fill-in.
\newblock {\em SIAM J. Comput.}, 42(6):2197--2216, 2013.

\bibitem{ghosh2012faster}
E.~Ghosh, S.~Kolay, M.~Kumar, P.~Misra, F.~Panolan, A.~Rai, and M.~Ramanujan.
\newblock Faster parameterized algorithms for deletion to split graphs.
\newblock In {\em Proceedings of the 13th Scandinavian Symposium and Workshops
  on Algorithm Theory (SWAT)}, volume 7357 of {\em Lecture Notes in Computer
  Science}, pages 107--118. Springer, 2012.

\bibitem{goldberg1995four}
P.~Goldberg, M.~Golumbic, H.~Kaplan, and R.~Shamir.
\newblock Four strikes against physical mapping of {DNA}.
\newblock {\em Journal of Computational Biology}, 2(1):139--152, 1995.

\bibitem{Golumbic80}
M.~C. Golumbic.
\newblock {\em Algorithmic Graph Theory and Perfect Graphs}.
\newblock Academic Press, New York, 1980.

\bibitem{ImpagliazzoPZ01}
R.~Impagliazzo, R.~Paturi, and F.~Zane.
\newblock Which problems have strongly exponential complexity?
\newblock {\em J. Comput. Syst. Sci.}, 63(4):512--530, 2001.

\bibitem{tarjan-sicomp}
H.~Kaplan, R.~Shamir, and R.~E. Tarjan.
\newblock Tractability of parameterized completion problems on chordal,
  strongly chordal, and proper interval graphs.
\newblock {\em SIAM J. Comput.}, 28(5):1906--1922, 1999.

\bibitem{KratschW13}
S.~Kratsch and M.~Wahlstr{\"o}m.
\newblock Two edge modification problems without polynomial kernels.
\newblock {\em Discrete Optimization}, 10(3):193--199, 2013.

\bibitem{cocoon13}
Y.~Liu, J.~Wang, C.~Xu, J.~Guo, and J.~Chen.
\newblock An effective branching strategy for some parameterized edge
  modification problems with multiple forbidden induced subgraphs.
\newblock In D.-Z. Du and G.~Zhang, editors, {\em COCOON}, volume 7936 of {\em
  Lecture Notes in Computer Science}, pages 555--566. Springer, 2013.

\bibitem{umbrella}
P.~J. Looges and S.~Olariu.
\newblock Optimal greedy algorithms for indifference graphs.
\newblock {\em Computers and Mathematics with Applications}, 25(7):15 -- 25,
  1993.

\bibitem{roberts}
F.~S. Roberts.
\newblock Indifference graphs.
\newblock In {\em Proof Techniques in Graph Theory: Proceedings of the Second
  Ann Arbor Graph Theory Conference}, pages 139 -- 146. Academic Press, New
  York, 1969.

\bibitem{Villanger:2009ez}
Y.~Villanger, P.~Heggernes, C.~Paul, and J.~A. Telle.
\newblock Interval completion is fixed parameter tractable.
\newblock {\em SIAM J. Comput.}, 38(5):2007--2020, 2009.

\bibitem{Yannakakis81}
M.~Yannakakis.
\newblock Computing the minimum fill-in is {NP}-complete.
\newblock {\em SIAM J. Alg. Disc. Meth.}, 2:77--79, 1981.

\end{thebibliography}

\end{document}